\documentclass[10pt,letterpaper]{article}
\usepackage[utf8]{inputenc}
\usepackage{amsmath}
\usepackage{amsfonts}
\usepackage{amssymb}
\usepackage{amsthm}
\usepackage{braket}
\usepackage{enumitem}
\usepackage{color}
\usepackage[sort,square,numbers]{natbib}
\usepackage{tikz}
\usetikzlibrary{quantikz}
\usepackage[basic,roman]{complexity}
\usepackage{subcaption}
\usepackage{makecell}
\usepackage{pifont}
\newcommand{\cmark}{{\color{green}\ding{51}}}%
\newcommand{\xmark}{{\color{red}\ding{55}}}%

\usepackage{appendix}
\usepackage{hyperref}
\usepackage[capitalize]{cleveref}
\usepackage[margin=1.25in]{geometry}

\usepackage{multirow}
\usepackage{cellspace} 
\usepackage{makecell} 
\setcellgapes{3pt}

\newtheorem{thm}{\protect\theoremname}
\theoremstyle{plain}
\newtheorem{lem}[thm]{\protect\lemmaname}
\theoremstyle{plain}
\newtheorem{rem}[thm]{\protect\remarkname}
\theoremstyle{plain}
\newtheorem*{lem*}{\protect\lemmaname}
\theoremstyle{plain}

\theoremstyle{plain}
\newtheorem{cor}[thm]{\protect\corollaryname}
\theoremstyle{plain}
\newtheorem{prob}[thm]{\protect\problemname}

\theoremstyle{plain}

\usepackage[USenglish]{babel}
  \providecommand{\problemname}{Problem}
  \providecommand{\corollaryname}{Corollary}
  \providecommand{\lemmaname}{Lemma}
  \providecommand{\propositionname}{Proposition}
  \providecommand{\remarkname}{Remark}
\providecommand{\theoremname}{Theorem}
\providecommand{\examplename}{Example}
\usepackage[normalem]{ulem}

\usepackage{algorithm}
\usepackage{algorithmic}

\newcommand{\CASE}[1]{\STATE \textbf{case} #1\textbf{:} \begin{ALC@g}}
\newcommand{\ENDCASE}{\end{ALC@g}}

\newcommand{\DEFAULT}{\STATE \textbf{default:} \begin{ALC@g}}
\newcommand{\ENDDEFAULT}{\end{ALC@g}}
\newcommand{\DEFAULTLINE}[1]{\STATE \textbf{default:} }
\newcommand{\Or}{\mathcal{O}}
\newcommand{\RR}{\mathbb{R}}

\newcommand{\ZZ}{\mathbb{Z}}
\newcommand{\wt}{\widetilde}
\newcommand{\Tr}{\mathrm{Tr}}

\newcommand{\dd}{\mathrm{d}}

\renewcommand{\Re}{\operatorname{Re}}
\renewcommand{\Im}{\operatorname{Im}}

\title{Heisenberg-limited ground state energy estimation for early fault-tolerant quantum computers}
\author{
Lin Lin\thanks{Department of Mathematics, and Challenge Institute of Quantum Computation,  University of California, Berkeley, and Computational Research Division, Lawrence Berkeley National Laboratory, Berkeley, CA 94720. }
\and
Yu Tong\thanks{Department of Mathematics, University of California, Berkeley,
CA 94720. }
}
\date{\today}

\begin{document}

\maketitle 

\begin{abstract}

Under suitable assumptions, the quantum phase estimation (QPE) algorithm is able to achieve Heisenberg-limited precision scaling in estimating the ground state energy. 
However, QPE requires a large number of ancilla qubits and large circuit depth, as well as the ability to perform inverse quantum Fourier transform, making it expensive to implement on an early fault-tolerant quantum computer. 
We propose an alternative method to estimate the ground state energy of a Hamiltonian with Heisenberg-limited precision scaling, which employs a simple quantum circuit with one ancilla qubit, and a classical post-processing procedure. Besides the ground state energy, our algorithm also produces an approximate cumulative distribution function of the spectral measure, which can be used to compute other spectral properties of the Hamiltonian.

\end{abstract}

\section{Introduction}
\label{sec:intro}

% For a Hamiltonian $H$ whose eigenpairs are $(\lambda_k,\ket{\psi_k})$, where $\lambda_k\leq\lambda_{k+1}$, we want to estimate $\lambda_0$ to within additive error $\epsilon$. We assume we can efficiently prepare an initial state $\ket{\phi}$. We denote $p_k=|\braket{\psi_k|\phi}|^2$.

Estimating the ground state energy of a quantum Hamiltonian is of immense importance in condensed matter physics, quantum chemistry, and quantum information. The problem can be described as follows:
we have a Hamiltonian $H$, acting on $n$ qubits, with the eigendecomposition
\[
H = \sum_{k=0}^{K{-}1} \lambda_k \Pi_k,
\]
where $\Pi_k$ is the projection operator into the $\lambda_k$-eigensubspace, and $\lambda_k$'s are increasingly ordered. Each eigenvalue may be degenerate, i.e. the rank of $\Pi_k$ can be more than one. We assume we can access the Hamiltonian $H$ through the time evolution operator $e^{-i\tau H}$ for some fixed $\tau$. Our goal is to estimate the ground state energy $\lambda_0$ to within additive error $\epsilon$.

Some assumptions are needed as otherwise this problem is \QMA-{hard} \cite{KitaevShenVyalyi2002, KempeKitaevRegev2006, OliveiraTerhal2005, AharonovGottesmanEtAl2009}. We assume we are given a state described by its density matrix $\rho$. Let $p_k = \Tr[\rho \Pi_k]$. Then if $p_0$ (i.e. the overlap between the initial state and the ground state) is reasonably large we can solve the ground state energy estimation problem efficiently. This assumption is reasonable in many practical settings. For example, in quantum chemistry, the Hartree-Fock method usually yields an approximate ground state that is easy to prepare on a quantum computer. {At least for relatively small molecular systems, the Hartree-Fock state can often have a large overlap with the exact ground state \cite{TubmanEtAl2018postponing}}. Therefore we may use the Hartree-Fock solution as $\rho$ in this setting. {Other candidates of $\rho$ that can be relatively easily prepared on quantum computers have been discussed in Refs.~\cite{BabbushMcCleanEtAl2015chemical,SugisakiEtAl2018quantum,TubmanEtAl2018postponing}, and an overview of methods to choose $\rho$ can be found in \cite[Section~V.A.2]{McardleEtAl2020quantum}.}

% Talk about Heisenberg limit
The computational complexity of this task depends on the desired precision $\epsilon$. Even in the ideal case where the exact ground state is given, this dependence cannot be better than linear in $\epsilon^{-1}$ for generic Hamiltonians \cite{AtiaAharonov2017}. This limit is called the Heisenberg limit \cite{GiovannettiLloydMaccone2006,GiovannettiLloydMaccone2011advances,ZwierzPerezDelgadoKok2010,ZwierzPerezDelgadoKok2012ultimate} in quantum metrology. This notion is closely related to the time energy uncertainty principle \cite{AharonovBohm1961time,AharonovMassarPopescu2002measuring,ChildsPreskillRenes2000quantum,AtiaAharonov2017}. This optimal scaling can be achieved using the quantum phase estimation (QPE) algorithm \cite{Kitaev1995}, which we will discuss in detail later.

% Talk about early fault-tolerant quantum computers
Much work has been done to develop the algorithms for ground state energy estimation both for near-term quantum devices \cite{peruzzo2014variational,mcclean2016theory,omalley2016scalable,Huggins2020nonorthogonalVQE}, and fully fault-tolerant quantum computers \cite{AbramsLloyd1999quantum, PoulinWocjan2009preparing, ge2019faster, lin2020near}. Relatively little work has been done for {early fault-tolerant quantum computers \cite{campbell2020early,BabbushMcCleanEtAl2021focus,BoothOGorman2021quantum,Layden2021first}} , which we expect to be able to accomplish much more complicated tasks than current and near-term devices, but still place significant limitations on the suitable algorithms. Refs.~\cite{kivlichan2020improved,campbell2020early} carried out careful resource cost estimation of performing QPE for the Hubbard model using surface code to perform quantum error correction. These are to our best knowledge the only works that addressed ground state energy estimation in the context of early fault-tolerant quantum computers.

{To be specific,} we expect such early fault-tolerant quantum computers to have the following characteristics: (1) The number of logical qubits are limited. (2) It is undesirable to have a large number of controlled operations. (3) It is a priority to reduce the circuit depth, e.g. it is better to run a circuit of depth $\Or(D)$ for $\Or(M)$ times than to run a circuit of depth $\Or(DM)$ for a constant number of times, even if using the shorter circuit entails some additional poly-logarithmic factors in the total runtime.

% The most prominent quantum algorithm for this task is quantum phase estimation (QPE). 
% Assuming the availability of the controlled Hamiltonian evolution operator $e^{-it H}$, then intuitively the ground state energy can be extracted via the phase information via the expectation value with respect to $\rho$ of the evolution operator at various $t$.  
% The Hamiltonian evolution operation can be viewed as a near-term algorithm\footnote{Quantum algorithms relying on multiple control qubits, large number of ancilla qubits, complex quantum subroutines, and/or long coherence time are generally considered to be beyond the scope of near-term technologies. On the other hand, quantum algorithms that use fewer control and ancilla qubits, relatively simple quantum operations and short coherence time, are loosely referred to as near-term algorithms. While they may still be beyond the reach of current noisy intermediate-scale quantum (NISQ) devices for problems of interest, they are nonetheless much more approachable in the near future. This paper focuses on such near-term algorithms for ground state energy estimation.} 
% thanks to methods such as Trotter splitting that can decompose a relatively complex Hamiltonian into simple gate rotations without involving ancilla qubits. 

In this context, the textbook version of QPE (see e.g. Refs.~\cite{CleveEkertEtAl1998quantum,nielsen2002quantum}), which uses multiple ancilla qubits to store the phase and relies on inverse quantum Fourier transform (QFT), 
% is not suitable for early fault-tolerant quantum computers. 
{has features that are not desirable on early fault-tolerant quantum computers.}
Some variants of QPE have been developed to achieve high confidence level \cite{knill2007optimal,poulin2009sampling,nagaj2009fast}, which can be important in many applications. However, such modifications require even more ancilla qubits to store multiple estimates of the phase and an additional coherent circuit to take perform logical operations.
%, and are not suitable for early fault-tolerant quantum computers. 
% \YT{Add a sentence to talk about resouce states}
{Another possible way to achieve high confidence level is to utilize a resource state (\cite[Section II B]{BabbushGidneyEtAl2018encoding}) to implement a Kaiser window filter \cite{SandersBerryEtAl2020compilation}. This approach requires the same number of ancilla qubits as the textbook version of QPE.}

Due to the above considerations, we focus on the variants of QPE that use only very few ancilla qubits (in fact, all algorithms below use only one ancilla qubit). Kitaev's algorithm (see e.g. \cite{KitaevShenVyalyi2002}) uses a simple quantum circuit with one control qubit to determine each bit of the phase individually. However this method, together with many other algorithms based on it \cite{wang2019accelerated,wiebe2015bayesian}, are designed for phase estimation with an eigenstate given exactly, which is different from our goal. The semi-classical Fourier transform \cite{griffiths1996semiclassical} can simulate QFT+measurement (meaning all qubits are measured in the end) with only one-qubit gates, classical control and post-processing, thus trading the expensive quantum resource for inexpensive classical operations. One can replace the inverse QFT with the semi-classical Fourier transform, and this results in a phase estimation algorithm that uses only one ancilla qubit \cite{HigginsBerryEtAl2007,BerryHiggins2009}. This approach can be seen as a simulation of the multiple-ancilla qubit version of QPE, and is therefore applicable to the case when $\rho$ is not exactly the ground state. Because of these attractive features this is the version of QPE used in Refs.~\cite{kivlichan2020improved,campbell2020early}. However, as we will explain below in \cref{sec:related}, this type of QPE requires running coherent time evolution for time $\Or(p_0^{-1}\epsilon^{-1})$. This leads to large circuit depth when $p_0$ is small. 
{Moreover, this approach cannot be used together with the resource state discussed earlier because the resource state is not a product state.}

In this work, the complexity is measured by the time for which we need to perform time evolution with the target Hamiltonian $H$. We will use two metrics: (1) the \emph{maximal evolution time}, which is the maximum length of time for which we need to perform (controlled) coherent time evolution, and (2) the \emph{total evolution time}, which is the sum of all the lengths of time we need to perform (controlled) coherent time evolution. They describe respectively the circuit depth and the total runtime. {Moreover, we will be primarily concerned with how they depend on the initial overlap $p_0$ and the precision $\epsilon$. The dependence on the system size $n$ mainly comes indirectly through $p_0$ and the conversion between the total evolution time and runtime, which we will discuss in more detail later.} 
%\YT{I added the word ``mainly''.} \LL{ check last sentence} \YT{there is a log dependence on $\tau^{-1}$ (Hamiltonian norm) but I think it's okay to ignore it here for simplicity.} 
We present an algorithm that achieves the following goals:
\begin{itemize}
\label{list:goals}
    % \item[(1)] Achieves Heisenberg-limited precision scaling, i.e. the total time for which we run time evolution depends linearly (up to a poly-logarithmic factor) on $\epsilon^{-1}$;
    \item[(1)] {Achieves Heisenberg-limited precision scaling, i.e. the total time for which we run time evolution is $\wt{\Or}(\epsilon^{-1}\poly(p_0^{-1}))$;}
    \item[(2)] Uses {at most} one ancilla qubit;
    \item[(3)] The maximal evolution time is at most $\Or(\epsilon^{-1}\polylog(\epsilon^{-1}p_0^{-1}))$.
\end{itemize} 
To our best knowledge our algorithm is the first to satisfy all three requirements. In our algorithm, we sample from a simple quantum circuit, and use the samples to approximately reconstruct the cumulative distribution function (CDF) of the spectral measure associated with the Hamiltonian. We then use classical post-processing to estimate the ground state energy with high confidence. Besides the ground state energy, our algorithm also produces the approximate CDF, which may be of independent interest. In the discussion above we assumed the controlled time evolution can be efficiently done. If controlled time evolution is  costly to implement, then based on ideas in Refs.~\cite{Huggins2020nonorthogonalVQE,RussoEtAl2020evaluating,LuBanulsCirac2020algorithms,OBrienEtAl2020error}, we offer an alternative circuit in Appendix~\ref{sec:control_free} which uses two ancilla qubits, with some additional assumptions.

{The problem of ground state energy estimation is closely related to that of ground state preparation, but there are important differences. First, having access to a good initial state $\rho$ (with large overlap with the ground state) does not make the energy estimation a trivial task, as even if we have access to the exact ground state the quantum resources required to perform phase estimation can still be significant. Second, ground state energy estimation algorithms do not necessarily involve ground state preparation. This is true for the algorithm in this work as well as in Refs.~\cite{ge2019faster,lin2020near}. Consequently, even though the ground state preparation algorithms generally have a runtime that depends on the spectral gap between the two lowest eigenvalues of the Hamiltonian, the cost of ground state energy estimation algorithms may not necessarily depend on the spectral gap.}

% \YT{Consider making it more ``coherent''}
{We remark that although we characterize the scaling as depending on the overlap $p_0$, in practice we need to know a lower bound of $p_0$, which we denote by $\eta$. The dependence on $p_0$ should more accurately be replaced by a dependence on $\eta$.}
{To our best knowledge, in order to obtain rigorous guarantee of the performance, the knowledge of $\eta$ (and that $\eta$ is not too small) is needed in all previous algorithms related to QPE. {This is because in QPE we need the knowledge of $\eta$ to obtain a stopping criterion. We will briefly explain this using a simple example.} Suppose we have a Hamiltonian $H$ on $n$ qubits with eigenvalues $\lambda_k$ (arranged in ascending order), and eigenstates $\ket{\psi_k}$, and $\ket{\phi_0}$ is an initial guess for the ground state. Furthermore we assume
$
p_0=|\braket{\phi_0|\psi_0}|^2 = 0.01,\  p_1 = |\braket{\phi_0|\psi_1}|^2 = 0.5.
$
We may idealize QPE as exact energy measurement to simplify discussion. If we have no a priori knowledge of $p_0$, then performing QPE on the state $\ket{\phi_0}$ will give us $\lambda_1$ with probability $1/2$. If we repeat this $\lesssim 100$ times most likely all energies we get will be $\geq \lambda_1$. Only when we measure $\gtrsim 100$ times can we reach the correct ground state energy $\lambda_0$. Hence if we do not know about a lower bound of $p_0$, we can never know whether we have stopped the algorithm prematurely.}

{The main idea of our algorithm is to use a binary search procedure to gradually narrow down the interval in which the ground state energy is located. The key component is a subroutine $\mathrm{CERTIFY}$ (Algorithm~\ref{alg:certify}) that distinguishes whether the ground state energy is approximately to the left or right of some given value. This, however, can only be perform up to certain precision, and can fail with non-zero probability. Therefore our search algorithm needs to account for this fuzzy outcome to produce a final result that is correct with probability arbitrarily close to $1$. In the $\mathrm{CERTIFY}$ procedure, we use a stochastic method to evaluate the cumulative distribution function associated with the spectral density, and this is the key to achieving the Heisenberg scaling. This stochastic method is described in detail in Section~\ref{sec:evaluate_ACDF}.}

\subsection{Related works}\label{sec:related}

We first briefly analyze the cost of the textbook version of QPE using multiple ancilla qubits. Although this method {has features that are not desirable on early fault-tolerant quantum computers,}
% is not suitable for early fault-tolerant quantum computers, 
this analysis will nevertheless be helpful for understanding the cost of other variants of QPE. For simplicity we assume $\rho=\ket{\phi}\bra{\phi}$ is a pure state, and the ground state $\ket{\psi_0}$ is non-degenerate. Approximately, the QPE  performs a projective measurement in the eigenbasis of $H$. With probability $p_0$, $\ket{\phi}$ will collapse to the ground state $\ket{\psi_0}$. If this happens the energy register will then give the ground state energy $\lambda_0$ to precision $\epsilon$. Therefore we run phase estimation for a total of $\Or(p_0^{-1})$ times, and take the instance with the minimum value in the energy register. With high probability this value will be close to $\lambda_0$. Each single run takes time $\Or(\epsilon^{-1})$. The total runtime cost is therefore $\Or(p_0^{-1}\epsilon^{-1})$. For simplicity here we do not consider the runtime needed to prepare $\ket{\phi}$.

\begin{figure}
    \centering
    \includegraphics[width=0.45\textwidth]{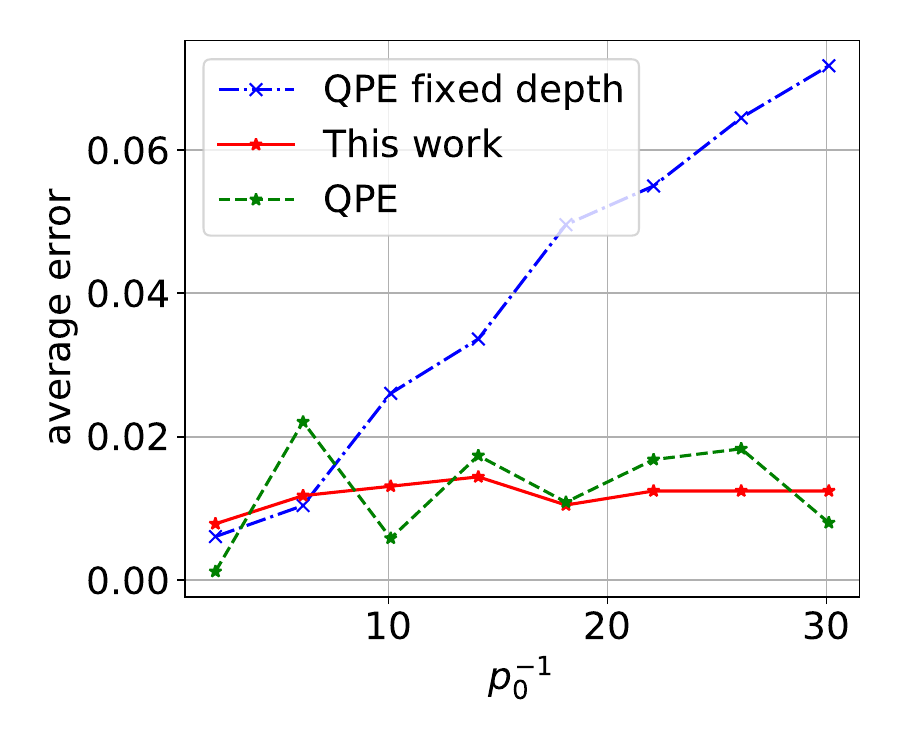}
    \includegraphics[width=0.45\textwidth]{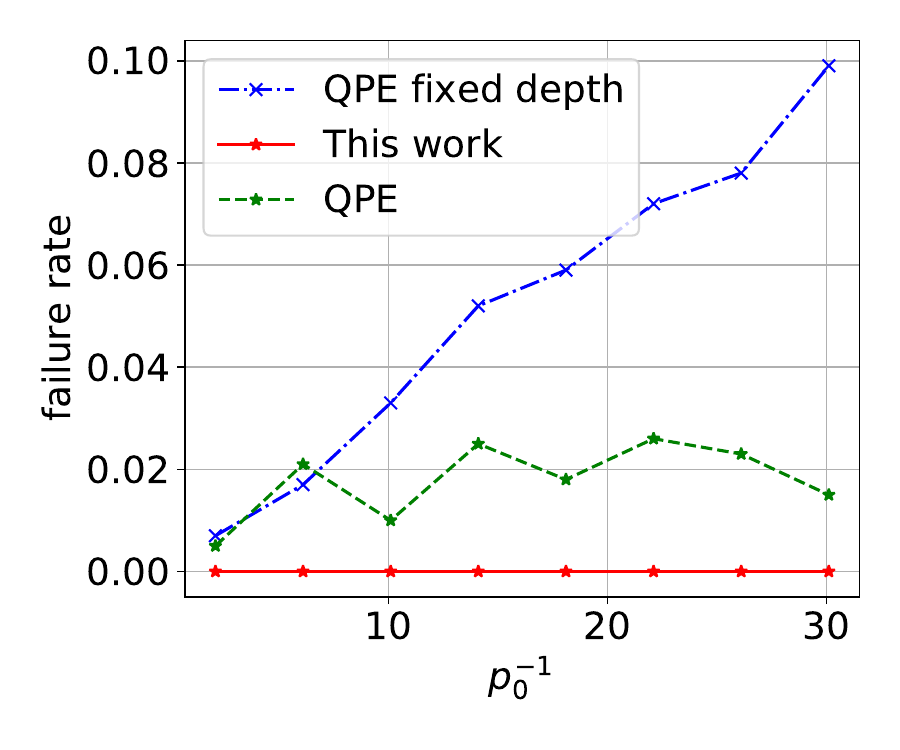}
    \caption{Comparing the performance of the textbook version QPE (blue dashed-dotted line) and the method in this work (red solid line) in ground state energy estimation with a fixed maximal evolution time ($300$ steps of time evolution with $H$) and decreasing initial overlap $p_0$. The results are benchmarked against QPE with maximal evolution time proportional to $p_0^{-1}$ (green dashed line). To use QPE, either with fixed or $\Or(p_0^{-1})$ maximal evolution time, to estimate the ground state energy, we run QPE for $\Or(p_0^{-1})$ times and take the minimum in energy measurement outcomes as the ground state energy estimate. The error is averaged over multiple runs, and the failure rate is the percentage of runs that yield an estimate with error larger than the tolerance $0.04$. The Hamiltonian $H$ is the Hubbard Hamiltonian defined in Eq.~\eqref{eq:hubbard_ham} with $U=10$, and the overlap $p_0$ is artificially tuned.}
    \label{fig:cf_QPE_filter}
\end{figure}

The above analysis, however, is overly optimistic. Since we need to repeat the phase estimation procedure for a total of $\Or(p_0^{-1})$ times, for an event that only has $\Or(p_0)$ probability of happening in a single run, the probability of this event occurring at least once in the total $\Or(p_0^{-1})$ repetitions is now $\Or(1)$ {(which means we cannot ensure that the error happens with sufficient low probability)}. {In our setting}, suppose the maximal evolution time is $T$, then each time we measure the energy register there is a $\Or(T^{-1}\epsilon'^{-1})$ probability that the output will be smaller than $\lambda_0-\epsilon'$. If we choose $T=\Or(\epsilon^{-1})$ as discussed above, and we let $\epsilon'=\epsilon/p_0$, then the probability of the minimum of the $\Or(p_0^{-1})$ energy register measurement outputs being smaller than $\lambda_0-\epsilon/p_0$ is {only upper bounded by} $\Or(1)${, and we can no longer control over the probability of the error being larger than $\epsilon$}. 
% Moreover, because the binary stored in the energy register output can deviate from the eigenvalues by $\epsilon'$ with probability $\Or(T^{-1}\epsilon'^{-1})$ where $T$ is the total time we run time evolution with $H$, if we choose $T=\Or(\epsilon^{-1})$ as discussed above, then the probability of an error of order $\epsilon/p_0$ occurring is $\Or(p_0)$. Therefore in $\Or(p_0^{-1})$ repetitions there is a constant probability that such a large error occurs at least once. 
This means there {might be} a high probability that the error of the ground state energy in the end will be of order $\epsilon/p_0$ instead of $\epsilon$. For a more formal analysis see \cite[Appendix A]{ge2019faster}.
We numerically demonstrate {that this is indeed the case} in Figure~\ref{fig:cf_QPE_filter}, in which we show the error increases as $p_0$ decreases and there is a larger probability of the estimate deviating beyond a prescribed tolerance if the maximal evolution time, or equivalently the circuit depth, for QPE is fixed.

To avoid this, one can instead choose the maximal evolution time to be  $T=\Or(p_0^{-1}\epsilon^{-1})$. After repeating $\Or(p_0^{-1})$ times, the total runtime then becomes $\Or(p_0^{-2}\epsilon^{-1})$. The increase in maximal evolution time can prevent the increase of error (see Figure~\ref{fig:cf_QPE_filter}). However, the extra $p_0^{-1}$ factor increases the circuit depth and  is undesirable. 

{There are several other algorithms based on phase estimation using a single ancilla qubit \cite{wang2019accelerated,wiebe2015bayesian,OBrienTarasinskiTerhal2019quantum} that are designed for  different settings from ours: they assume the availability of an exact eigenstate, or are designed for obtaining the entire spectrum and thus only work for small systems. Ref.~\cite{SommaOrtizEtAl2002simulating} proposes a method for estimating the eigenvalues by first estimating $\Tr[\rho e^{-it H}]$ and then performing a classical Fourier transform, but no runtime scaling is provided.}
The semi-classical Fourier transform \cite{griffiths1996semiclassical} simulates the QFT in a classical manner, and the QPE using single ancilla qubit and semi-classical Fourier transform has the same scaling in terms of the maximal evolution time and the total evolution time.

In order to improve the dependence on $p_0$,  we may use the high-confidence versions of the phase estimation algorithm \cite{knill2007optimal,poulin2009sampling,nagaj2009fast}. 
In this method, the maximal evolution time required can be reduced to $\Or(\epsilon^{-1}\log(p_0^{-1}))$, through taking the median of several copies of the energy register in a coherent manner. However, this requires
% multiplying the circuit width by a factor $\Or(\log(\nu^{-1}))$, 
using multiple copies of the energy register,
together with an additional quantum circuit to compute the medians coherently that can be difficult to implement. 
Note that semi-classical Fourier transform can only simulate the measurement outcome and does not preserve coherence, 
and therefore to our knowledge, the high-confidence version of phase estimation cannot be modified to use only a single qubit. 
In Ref.~\cite{ge2019faster}, the authors used a method called minimum label finding to improve the runtime to $\Or(p_0^{-3/2}\epsilon^{-1})$, but the implementation of the minimum label finding with limited quantum resources is again difficult.

Besides these algorithms based on phase estimation, several other algorithms have been developed to solve the ground state energy problem. Ref.~\cite{ge2019faster} proposed a method based on the linear combination of unitaries (LCU) technique that requires running time evolution for duration $\wt{\Or}(p_0^{-1/2}\epsilon^{-3/2})$ and preparing the initial state  $\wt{\Or}(p_0^{-1/2}\epsilon^{-1/2})$ times.\footnote{In this paper we use the following asymptotic notations besides the usual $\Or$ notation: we write $f=\Omega(g)$ if $g=\Or(f)$; $f=\Theta(g)$ if $f=\Or(g)$ and $g=\Or(f)$; $f=\wt{\Or}(g)$ if $f=\Or(g\operatorname{polylog}(g))$.} Assuming the Hamiltonian $H$ is available in its block-encoding \cite{low2019hamiltonian,chakraborty2018power}, Ref.~\cite{lin2020near} uses quantum signal processing \cite{low2017optimal,gilyen2019quantum} with a binary search procedure, which queries the block-encoding $\wt{\Or}(p_0^{-1/2}\epsilon^{-1})$ times and prepares the initial state  $\wt{\Or}(p_0^{-1/2}\log(\epsilon^{-1}))$ times. To our knowledge, this is the best complexity that has been achieved. However the block-encoding of a quantum Hamiltonian of interest, LCU, and amplitude estimation techniques (used in \cite{lin2020near}) are expensive in terms of the number of ancilla qubits, controlled operations, and logical operations needed.
% generally considered to be beyond near-term technologies. 

A very different type of algorithms for ground state energy estimation is the variational quantum eigensolver (VQE) \cite{peruzzo2014variational,mcclean2016theory,omalley2016scalable}, which are near-term algorithms and have been demonstrated on real quantum computers. The accuracy of VQE is limited both by the representation power of the variational ansatz, and the capabilities of classical optimization algorithms for the associated non-convex optimization problem. Hence unlike aforementioned algorithms, there is no provable performance guarantees for VQE-type methods. In fact some recent results show solving the non-convex optimization problem can be \NP-hard \cite{bittel2021VQE_NPhard}. Furthermore, each evaluation of the energy expectation value to precision $\epsilon$ requires $\Or(\epsilon^{-2})$ samples due to Monte Carlo sampling. This can to some extent be remedied using the methods in \cite{knill2007optimal,wang2019accelerated} at the expense of larger circuit depth requirement. 

There are also a few options that can be viewed to be in-between VQE and QPE. The quantum imaginary time evolution (QITE) algorithm \cite{Motta2019QITE} uses state tomography turning an imaginary time evolution into a series of real time Hamiltonian evolution problem. Inspired by the classical Krylov subspace method, Refs.~\cite{stair2020multireference,parrish2019quantum,Huggins2020nonorthogonalVQE} propose to solve the ground state energy problem by restricting the Hilbert space to a low dimension space spanned by some eigenstates that are accessible with time evolution. Similar to VQE, no provable complexity upper bound is known for these algorithms, and all algorithms suffer from the $\epsilon^{-2}$ scaling due to the Monte Carlo sampling. In fact, the stability of these algorithms remains unclear in the presence of sampling errors.

A more ambitious goal than ground state energy estimation is to estimate the distribution of all eigenvalues weighted by a given initial state $\rho$ {\cite{EmersonLloydEtAl2004estimation,OBrienTarasinskiTerhal2019quantum,somma2019quantum}}. Using a quantum circuit similar to that in Kitaev's algorithm as well as classical post-processing, Ref.~\cite{somma2019quantum} proposed an algorithm to solve the quantum eigenvalue estimation problem (QEEP). We henceforth refer to this algorithm as the quantum eigenvalue estimation algorithm (QEEA). Suppose $\|H\|\leq 1/2$, and the interval $[-\pi,\pi]$ is divided into $M$ bins of equal size denoted by $B_j=[-1/2+ j/M,-1/2+(j+1)/M]$. Then QEEA estimates the quantities $q_j=\sum_{k:\lambda_k\in B_j}p_k$. Although QEEA was not designed for ground state energy estimation, one can use this algorithm to find the leftmost bin in which $q_j\geq p_0/2$, and thereby locate the ground state energy within a bin of size $M^{-1}$. While the maximal evolution time required scales as $\Or(\epsilon^{-1})$, the total evolution time of the original QEEA scales as $\Or(\epsilon^{-6})$. We analyze the cost of QEEA in Appendix~\ref{sec:qeep}, and show that the total runtime can be reduced to $\Or(\epsilon^{-4})$ for the ground state energy estimation in a straightforward way, yet this is still costly if high precision is required.

% \LL{modify this sentence consistently, maybe before the "related works" section} In all estimates above, the circuit depth and the total runtime are measured in terms of the maximum and the total time of the Hamiltonian evolution respectively. % The cost and the error for the Hamiltonian evolution are orthogonal to the problem of post-processing the data obtained from the quantum circuit, and are assumed to be systematically controlled. 
% The cost and the error for the Hamiltonian evolution are orthogonal to the problem of ground state energy estimation and can be controlled separately.  

To the extent of our knowledge, none of the existing algorithms achieves all three goals listed on Page~\pageref{list:goals}. Some can have better maximal evolution time or total evolution time requirement, but the advantage always comes at the expense of some other aspects. {In Table~\ref{tab:compare_algs_full} we list the quantum algorithms discussed in this work and whether they satisfy each of the requirements.}

% \begin{table}[th]%\fontsize{7}{8}
%     \centering
%     \makegapedcells
%         \begin{tabular}{p{6cm}|c}
%         \hline
%         \hline
%         Algorithm   & Requirement(s) not satisfied \\
%         \hline
%         QPE (textbook version) \cite{CleveEkertEtAl1998quantum,nielsen2002quantum}        &        (2), (3)  \\
%         QPE (high-confidence) \cite{knill2007optimal,poulin2009sampling,nagaj2009fast}          &       (2)     \\
%         QPE (semi-classical QFT)  \cite{HigginsBerryEtAl2007,BerryHiggins2009}       &      (3)  \\
%         QPE (Iterative)  \cite{KitaevShenVyalyi2002}                &      Needs exact eigenstate \\
%         The LCU approach  \cite{ge2019faster}               &      (2) \\
%         The binary search approach \cite{lin2020near}      &      (2) \\
%         VQE \cite{peruzzo2014variational,mcclean2016theory,omalley2016scalable}                           &      (1), no precision guarantee \\
%         QITE \cite{Motta2019QITE}                            &      (1), no precision guarantee \\
%         QEEA \cite{somma2019quantum}                            &      (1) \\
%         Krylov subspace methods  \cite{stair2020multireference,parrish2019quantum,Huggins2020nonorthogonalVQE}        &      (1), no precision guarantee \\
%         \hline 
%         \hline
%         \end{tabular}
%     \caption{{Comparison of quantum algorithms for estimating the ground state energy, with respect to the requirements on Page~\pageref{list:goals}.} \LL{ Maybe change the format (see slack message)} }
%     \label{tab:compare_algs_full}
% \end{table}

\begin{table}[th]%\fontsize{7}{8}
    \centering
    \makegapedcells
        \begin{tabular}{p{6cm}|ccc|l}
        \hline
        \hline
        \multirow{2}{*}{Algorithms}   & \multicolumn{3}{c|}{Requirements} & \multirow{2}{*}{Other issues} \\
         & (1) & (2) & (3) \\
        \hline
        QPE (textbook version) \cite{CleveEkertEtAl1998quantum,nielsen2002quantum} & \cmark & \xmark & \xmark &  \\
        QPE (high-confidence) \cite{knill2007optimal,poulin2009sampling,nagaj2009fast} & \cmark & \xmark & \cmark &  \\
        QPE (semi-classical QFT)  \cite{HigginsBerryEtAl2007,BerryHiggins2009}  &  \cmark & \cmark & \xmark &  \\
        QPE (iterative)  \cite{KitaevShenVyalyi2002}   & \cmark & \cmark & \cmark &   Needs exact eigenstate ($p_0=1$) \\
        The LCU approach  \cite{ge2019faster}   &  \xmark & \xmark & \xmark &   \\
        The binary search approach \cite{lin2020near}  & \cmark & \xmark & \xmark & \\
        VQE \cite{peruzzo2014variational,mcclean2016theory,omalley2016scalable}                           & \xmark & \cmark & ? & No precision guarantee \\
        QITE \cite{Motta2019QITE}   & \xmark & \cmark & ? & Requires state tomography\\
        QEEA \cite{somma2019quantum} & \xmark & \cmark & \cmark & \\
        Krylov subspace methods \cite{stair2020multireference,parrish2019quantum,Huggins2020nonorthogonalVQE}   & \xmark & \cmark & ?  & No precision guarantee \\
        This work & \cmark & \cmark & \cmark & \\
        \hline 
        \hline
        \end{tabular}
    % \caption{{Comparison of quantum algorithms for estimating the ground state energy, with respect to the requirements on Page~\pageref{list:goals}.} \LL{ Maybe change the format (see slack message)} }
    \caption{{Quantum algorithms for estimating the ground state energy and whether they satisfy each of the three requirements on Page~\pageref{list:goals}. We recall that the requirements are (1) achieving the Heisenberg-limited precision scaling, (2) using at most one ancilla qubit, and (3) the maximal evolution time being at most $\Or(\epsilon^{-1}\polylog(\epsilon^{-1}p_0^{-1}))$.}}
    \label{tab:compare_algs_full}
\end{table}

In Table~\ref{tab:compare_algs}, we compare the maximal evolution time, the number of repetitions (the number of times we need to run the quantum circuit), and the total evolution time needed, using the three qubit-efficient methods that require only one ancilla qubit.
%\LL{ Is it worth comparing the performance of ``best runners-up'' of near-term algorithms with the new algorithm in a table?}  

\begin{table}[th]%\fontsize{7}{8}
    \centering
    \makegapedcells
        \begin{tabular}{p{4cm}|c|c|c}
        \hline
        \hline
                    & Max evolution time & Repetitions & Total evolution time \\
        \hline
        This work (Corollary~\ref{cor:ground_energy})   & {$\wt{\Or}(\epsilon^{-1}\polylog(p_0^{-1}))$} & {$\wt{\Or}(p_0^{-2}\polylog(\epsilon^{-1}))$} & {$\wt{\Or}(\epsilon^{-1}p_0^{-2})$} \\
        \hline
        QPE with semi-classical Fourier transform    & {$\wt{\Or}(\epsilon^{-1}p_0^{-1})$} & {$\wt{\Or}(p_0^{-1}\polylog(\epsilon^{-1}))$} & {$\wt{\Or}(\epsilon^{-1}p_0^{-2})$} \\
        \hline
        QEEA  \cite{somma2019quantum} & {$\wt{\Or}(\epsilon^{-1}\polylog(p_0^{-1}))$} & {$\wt{\Or}(\epsilon^{-3}p_0^{-2})$} & {$\wt{\Or}(\epsilon^{-4}p_0^{-2})$} \\
        \hline 
        \hline
        \end{tabular}
    \caption{Comparison of  the maximal evolution time, the number of repetitions (the number of times we need to run the quantum circuit), and the total evolution time needed for estimating the ground state energy to within error $\epsilon$, using the three methods that require only one ancilla qubit: the method in this work, QPE with semi-classical Fourier transform that uses only one ancilla qubit, and the QEEA in Ref.~\cite{somma2019quantum}. The overlap between the initial state and the ground state is assumed to be $p_0$. 
The number of repetitions is also the number of times we need to prepare the initial state. An analysis of the QEEA in Ref.~\cite{somma2019quantum} can be found in Appendix~\ref{sec:qeep}.}
    \label{tab:compare_algs}
\end{table}

% \begin{table}[th]%\fontsize{7}{8}
%     \centering
%     \makegapedcells
%         \begin{tabular}{|p{4cm}|c|c|c|}
%         \hline
%         \hline
%                     & Max evolution time & Repetitions & Total evolution time \\
%         \hline
%         This work (Corollary~\ref{cor:ground_energy})   & $\Or(\epsilon^{-1})$ & $\Or(p_0^{-2})$ & $\Or(\epsilon^{-1}p_0^{-2})$ \\
%         \hline
%         QPE with semi-classical Fourier transform    & $\Or(\epsilon^{-1}p_0^{-1})$ & $\Or(p_0^{-1})$ & $\Or(\epsilon^{-1}p_0^{-2})$ \\
%         \hline
%         QEEA  \cite{somma2019quantum} & $\Or(\epsilon^{-1})$ & $\Or(\epsilon^{-3}p_0^{-2})$ & $\Or(\epsilon^{-4}p_0^{-2})$ \\
%         \hline 
%         \hline
%         \end{tabular}
%     \caption{Comparison of  the maximal evolution time, the number of repetitions (the number of times we need to run the quantum circuit), and the total evolution time needed for estimating the ground state energy to within error $\epsilon$, using the three methods that require only one ancilla qubit: the method in this work, QPE with semi-classical Fourier transform that uses only one ancilla qubit, and the QEEA in Ref.~\cite{somma2019quantum}. The overlap between the initial state and the ground state is assumed to be $p_0$. 
% The number of repetitions is also the number of times we need to prepare the initial state. An analysis of the QEEA in Ref.~\cite{somma2019quantum} can be found in Appendix~\ref{sec:qeep}. For clarity we have omitted all poly-logarithmic dependence on $\epsilon^{-1}$ and $p_0^{-1}$.}
%     \label{tab:compare_algs}
% \end{table}

Finally, in a gate-based setting, the exact relations between the maximal evolution time and the circuit depth, and between the total evolution time and the total runtime, can be affected by the method we use to perform time evolution. Suppose we have access to a unitary circuit that performs $e^{-i\tau H}$ exactly for some fixed $\tau$. Then in order to run coherent time evolution for time $T$ we only need to use a circuit of depth $\Or(T)$. Therefore the circuit depth scales linearly with respect to the maximal evolution time. Similarly the total runtime scales linearly with respect to the total evolution time.

However, if we can only perform time evolution through Hamiltonian simulation, then these relations become more complicated. If advanced Hamiltonian simulation methods \cite{low2017optimal,low2019hamiltonian,BerryChildsKothari2015} can be used, the additional cost would be asymptotically negligible, since to ensure an $\epsilon'$ error for time evolution for time $T$ the cost is $\Or(T\polylog(T\epsilon'^{-1}))$. Hence the cost is only worse than that in the ideal case by a poly-logarithmic factor. However, for early fault-tolerant quantum computers, as discussed in Refs.~\cite{kivlichan2020improved,campbell2020early}, Trotter formulas \cite{Suzuki1991} are generally favored. Running time evolution for time $T$ with error at most $\epsilon'$ would entail a runtime of $\Or(T^{1+1/p}\epsilon'^{-1/p})$.
The additional cost will therefore prevent us from reaching the Heisenberg limit, though high-order Trotter formulas (i.e. with a large $p$) can allow us to get arbitrarily close to the Heisenberg limit. If one does not insist on having a Heisenberg-limited scaling, then randomized algorithms \cite{Campbell2019random,BerryChildsSuWangWiebe2020,ChenHuangKuengTropp2020quantum} may lead to lower gate count when only low precision is required.

In Appendix~\ref{sec:trotter} we analyze the circuit depth and the total runtime of our algorithm with time evolution performed using Trotter formulas. We also compare with QPE based on Trotter formulas. We found that when using Trotter formulas, our method has some additional advantage over QPE, achieving a polynomially better dependence on $p_0$ (i.e. $\eta$ in Appendix~\ref{sec:trotter}) in the total runtime. The total runtime scales like $\epsilon^{-1-o(1)}$ using our algorithm with Trotter formulas, and this only approximately reaches the Heisenberg limit $\epsilon^{-1}$ in terms of the total runtime. However, it is worth noting that none of the other methods can strictly reach the Heisenberg limit using Trotter formulas. Otherwise we can instead perform Hamiltonian simulation with the exponentially accurate methods to go below the Heisenberg limit, which is an impossible task. {Despite the sub-optimal asymptotic scaling, with tight error analysis \cite{ChildsSuTranWiebeZhu2021,TranChuEtAl2020destructive,ChildsSu2019nearly,YiCrosson2021spectral} Trotter formulae may outperform the advanced Hamiltonian simulation techniques discussed above in terms of the gate complexity, especially when only moderate accuracy is needed.}

\subsection{Organization}
\label{sec:organization}
The rest of the paper is organized as follows. In Section~\ref{sec:overview} we introduce the quantum circuit we are going to use, and introduce the CDF which is going to play an important role in our algorithm, and give an overview of the ground state energy estimation algorithm. In Section~\ref{sec:evaluate_ACDF} we discuss how to approximate the CDF. In Section~\ref{sec:find_ground_energy} we show that the ground state energy can be estimated by inverting the CDF, and present the complexity of our algorithm (Corollary~\ref{cor:ground_energy}). In Section~\ref{sec:invert_CDF} we present the details of our algorithm for post-processing the measurement data and analyze the complexity.

\section{Overview of the method}
\label{sec:overview}

We want to keep the quantum circuit we use as simple as possible. In this work we use the following circuit
\begin{equation}
\label{eq:Hadamard_circuit}
\begin{quantikz}
\lstick{$\ket{0}$} & \gate{\mathrm{H}}    & \ctrl{1}             & \gate{W} & \gate{\mathrm{H}} & \meter{} \\
\lstick{$\rho$}    & \qwbundle{} & \gate{e^{-ij\tau H}} & \qw      & \qw      & \qw  \\
\end{quantikz}
\end{equation}
where $\mathrm{H}$ is the Hadamard gate. We choose $W=I$ or $W=S^{\dagger}$ where $S$ is the phase gate, depending on the quantity we want to estimate. The quantum circuit is simple and uses only one ancilla qubit as required. The quantum circuit itself has been used in previous methods~\cite{KitaevShenVyalyi2002,somma2019quantum}. However, our algorithm uses a different strategy for querying the circuit and for classical post-processing, and results in lower total evolution time and/or maximal evolution time achieving the goals (1) and (3) listed on Page~\pageref{list:goals}.

This circuit requires controlled time evolution, which can be non-trivial to implement. {The idea of removing controlled operation in phase estimation has also been considered in \cite{BoixoSomma2008parameter}.} {Here we} can use ideas from Refs.~\cite{Huggins2020nonorthogonalVQE,LuBanulsCirac2020algorithms,RussoEtAl2020evaluating,OBrienEtAl2020error} to remove the need to perform controlled time evolution. But this type of approach requires an eigenstate of $H$ with known eigenvalue that is easy to prepare. In a second-quantized setting we can simply use the vacuum state. We will discuss this in detail in Appendix~\ref{sec:control_free}.

Using the circuit in \eqref{eq:Hadamard_circuit}, in order to estimate $\Re\Tr[\rho e^{-ij\tau H}]$, {where $j$ is an arbitrary integer and $\tau$ is a real number,} we set $W=I$. We introduce a random variable $X_j$ and set it to be $1$ when the measurement outcome is $0$, and $-1$ when the measurement outcome is $1$. Then
\begin{equation}
    \label{eq:measurement_outcome_expect_X}
    \mathbb{E}[X_j] = \Re\Tr[\rho e^{-ij\tau H}].
\end{equation}
Similarly for $\Im\Tr[\rho e^{-ij\tau H}]$, we set $W=S^{\dagger}$, and introduce a random variable $Y_j$ that depends in the same way on the measurement outcome. We have
\begin{equation}
    \label{eq:measurement_outcome_expect_Y}
    \mathbb{E}[Y_j] = \Im\Tr[\rho e^{-ij\tau H}].
\end{equation}
The parameter $\tau$ is chosen to normalize the Hamiltonian. Specifically, we choose $\tau$ so that $\tau\|H\|<\pi/3$. %{We remark that $\tau$ should be chosen to be $\Theta(\|H\|^{-1})$ to avoid unnecessary overheads.}
{We remark that $\tau$ should be chosen to be $\Or(\|H\|^{-1})$, and to avoid unnecessary overheads we want its scaling to be as close to $\Theta(\|H\|^{-1})$ as possible.} 

% We have a Hamiltonian $H$ whose eigenpairs are $(\lambda_k,\ket{\psi_k})$, with $\lambda_k\leq \lambda_{k+1}$. The goal is to estimate $\lambda_0$ to within some additive error $\epsilon$. We consider the simplest quantum circuit, in which we have access to an initial state specified by the density matrix $\rho$, and the circuit performs a Hadamard test, and upon measurement yielding outcomes $X,Y\in\{\pm 1\}$ with probability specified by 
% \begin{equation}
% \label{eq:XY_dist}
% \Pr[X=1] = \frac{1}{2}(1+\Re\Tr[\rho e^{-ij\tau H}]),\quad \Pr[Y=1] = \frac{1}{2}(1+\Im\Tr[\rho e^{-ij\tau H}]).
% \end{equation}
% Here $\tau$ is chosen so that $\tau\|H\|< \pi/3$.

We can define a spectral measure of $\tau H$ associated with $\rho$. The spectral measure is
\begin{equation}
    \label{eq:spectral_measure_defn}
    p(x) = \sum_{k=0}^{K-1} p_k \delta(x-\tau\lambda_k),\quad x\in[-\pi,\pi].
\end{equation}
{Here $K$ is the number of different eigenvalues,  $\lambda_k$'s are the distinct eigenvalues arranged in ascending order, and each $p_k$ is the corresponding overlap, as defined in the Introduction.}
We extend it to a $2\pi$-periodic function by $p(x+2\pi)=p(x)$ 
{so that the Fourier transform can be performed on the interval $[0,2\pi]$ instead of the whole real line, which leads to a discrete Fourier spectrum.}
% in order to make the discussion below more convenient. 
Note that because of the assumption $\tau\|H\|< \pi/3$, within the interval $[-\pi,\pi]$, $p(x)$ is supported in $(-\pi/3,\pi/3)$. Next we consider the cumulative distribution function (CDF) associated with this measure.

We define the $2\pi$-periodic Heaviside function by
\begin{equation}
\label{eq:perodic_heaviside}
H(x)=\begin{cases}
1,\ x\in[2k\pi,(2k+1)\pi) ,\\
0,\ x\in[(2k-1)\pi,2k\pi),
\end{cases}
\end{equation}
where $k\in\ZZ$.
The CDF is usually defined by
$
C(x) = \sum_{k:\lambda_k\leq x} p_k.
$
This is however not a $2\pi$-periodic function and thus will create technical difficulties in later discussions. Therefore instead of the usual definition, we define
\begin{equation}
\label{eq:CDF_def}
C(x) =  (H*p)(x),
\end{equation}
% \begin{equation}
% \label{eq:CDF_def}
% C(x) = \int_{-\pi}^x p(y)\dd y = (H*p)(x).
% \end{equation}
where $*$ denotes convolution. There is ambiguity at the jump discontinuities, and we define the values of $C(x)$ at these points by requiring $C(x)$ to be right-continuous.
% The reason we choose to integrate from $-\pi$ is that we do not want any eigenvalue to be double-counted because of the periodic extension. 
% One can 
{We} check that this definition agrees with the usual definition when $x\in (-\pi/3,\pi/3)$, which is the interval that contains all the eigenvalues of $\tau H${:}
{
\begin{equation*}
\begin{aligned}
C(x) &= \int_{-\pi}^{\pi}H(y)p(x-y)\dd y 
= \int_{0}^{\pi} p(x-y)\dd y \\
&= \int_{x-\pi}^x p(y)\dd y 
 = \int_{-\pi}^{x} p(y)\dd y 
= \sum_{k:\lambda_k\leq x} p_k.
\end{aligned}
\end{equation*}
}
{Consequently} $C(x)$ is a right-continuous non-decreasing function in $(-\pi/3,\pi/3)$.

If we could evaluate the CDF then we would be able to locate the ground state energy. 
{This is because the CDF is a piecewise constant function. Each of its jumps in the interval $(-\pi/3,\pi/3)$  corresponds to an eigenvalue of $\tau H$. In order to find the ground state energy we only need to find where $C(x)$ jumps from zero to a non-zero value. However, in practice we cannot evaluate the CDF exactly.}
We will see that we are able to approximate, in a certain sense as will be made clear later, the CDF using a function we call the approximate CDF (ACDF). To this end we first define an approximate Heaviside function $F(x)=\sum_{|j|\leq d}\hat{F}_j e^{ijx}$ such that \begin{equation}
\label{eq:heaviside_approx_criterion_main_text}
|F(x)-H(x)|\leq \epsilon,\quad  x \in [-\pi+\delta,-\delta]\cup[\delta,\pi-\delta].    
\end{equation} 
The construction of this function is provided in Lemma~\ref{lem:approx_heaviside_function}{, where $\hat{F}_j$ is written as $\hat{F}_{d,\delta,j}$. Here the parameters $d$ and $\delta$ need to be chosen to control the accuracy of this approximation, and their choices will be discussed later}. We {also} omit the $d$ and $\delta$ dependence in the {subscripts} for simplicity. With this $F(x)$ we define the ACDF by
\begin{equation}
\label{eq:ACDF_definition}
\wt{C}(x) = (F*p)(x).
\end{equation}
In Section~\ref{sec:evaluate_ACDF} we will discuss how to evaluate this ACDF using the circuit in \eqref{eq:Hadamard_circuit}.
% We can evaluate it through
% \begin{equation}
% \label{eq:evaluate_ACDF}
% \begin{aligned}
% \wt{C}(x) &= \sum_{|j|\leq d} \hat{F}_j \int_{-\pi}^{\pi}p(y)e^{ij(x-y)}\dd y \\
% &= \sum_{|j|\leq d}\hat{F}_j e^{ijx} \Tr[\rho e^{-ij\tau H}],
% \end{aligned}
% \end{equation}
% where each $\Tr[\rho e^{-ij\tau H}]$ is evaluated through Monte Carlo sampling. 
The ACDF and CDF are related through the following inequalities
\begin{equation}
    \label{eq:relation_CDF_ACDF}
    C(x-\delta)-\epsilon \leq \wt{C}(x) \leq C(x+\delta)+\epsilon
\end{equation}
for any $|x|\leq \pi/3$, $0<\delta<\pi/6$ and $\epsilon>0$. We prove these inequalities in Appendix~\ref{sec:relation_CDF_ACDF}. {
Given the statistical estimation of the ACDF $\wt{C}(x)$, these inequalities enable us to estimate where the jumps of the CDF occur, which leads to an estimate of the ground state energy.}

By approximately evaluating the ACDF $\wt{C}(x)$ for certain chosen $x$, and through \cref{eq:relation_CDF_ACDF}, we can perform a binary search to locate the ground state energy in smaller and smaller intervals. The algorithm to do this and the total computational cost required to estimate the ground state energy to precision $\epsilon$ at a confidence level $1-\vartheta$ are discussed in Sections~\ref{sec:find_ground_energy} and \ref{sec:invert_CDF}.

\section{Evaluating the ACDF}
\label{sec:evaluate_ACDF}

\begin{figure}
    \centering
    \includegraphics[width=0.5\textwidth]{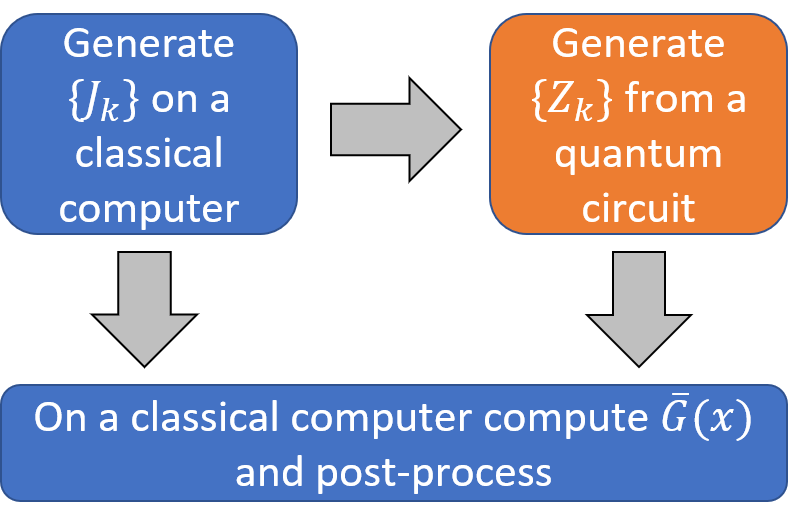}
    \caption{An illustration of the classical and quantum components of our algorithm: (1) generate samples $\{J_k\}$ from \eqref{eq:J_dist}; (2) use $\{J_k\}$ to generate $\{Z_k\}$ according to \eqref{eq:def_Z}; (3) compute $\bar{G}(x)$ through \eqref{eq:G_average}. The ground state energy estimate can be obtained through post-processing as discussed in Section~\ref{sec:find_ground_energy}. Only Step (2) needs to be performed on a quantum computer.}
    \label{fig:flow_chart}
\end{figure}

In this section we discuss how to evaluate the ACDF $\wt{C}(x)$. We first expand it in the following way:
\begin{equation}
\label{eq:evaluate_ACDF}
\begin{aligned}
\wt{C}(x) &= \sum_{|j|\leq d} \hat{F}_j \int_{-\pi}^{\pi}p(y)e^{ij(x-y)}\dd y \\
&= \sum_{|j|\leq d}\hat{F}_j e^{ijx} \Tr[\rho e^{-ij\tau H}],
\end{aligned}
\end{equation}
{where the spectral measure $p(x)$ is defined in \eqref{eq:spectral_measure_defn}. In going from the first line to the second line in the above equation we have used the fact that
\[
\int_{-\pi}^{\pi}p(y)e^{-ijy}\dd y = \sum_{k=0}^{K-1}\Tr[\rho\Pi_k]e^{-ij\tau\lambda_k}=\Tr[\rho e^{-ij\tau H}].
\]
}One might want to evaluate each $\Tr[\rho e^{-ij\tau H}]$ using Monte Carlo sampling since this quantity is equal to $\mathbb{E}[X_j+iY_j]$. If we want to evaluate all $\Tr[\rho e^{-ij\tau H}]$ to any accuracy at all, we need to sample each $X_j$ and $Y_j$ at least once. Then the total evolution time is is at least $\tau\sum_{|j|\leq d}|j|=\Omega(\tau d^2)$. 
Later we will see we need to choose $d=\Or(\epsilon^{-1}\polylog(\epsilon^{-1}p_0^{-1}))$ to ensure the ground state energy estimate has an additive error smaller than $\epsilon$. Hence this total evolution time would give rise to a $\epsilon^{-2}$ dependence in the runtime.
%Later we will see this leads to a $\epsilon^{-2}$ dependence on the additive error $\epsilon$ of the ground state energy. \LL{ did this appear later?} 

In order to avoid this $\epsilon^{-2}$ dependence, instead of evaluating all the terms we stochastically evaluate \eqref{eq:evaluate_ACDF} as a whole. The idea we are going to describe is {inspired by} the unbiased version of the multi-level Monte Carlo method \cite{rhee2012new,rhee2015unbiased}. We define a random variable $J$ that is drawn from $\{-d,-d+1,\ldots, d\}$, with probability
\begin{equation}
    \label{eq:J_dist}
    \Pr[J=j] = \frac{|\hat{F}_j|}{\mathcal{F}} ,
\end{equation}
where the normalization factor $\mathcal{F}=\sum_{|j|\leq d}|\hat{F}_j|$. We let $\theta_j$ be the argument of $\hat{F}_j$, i.e. $\hat{F}_j = |\hat{F}_j|e^{i\theta_j}$. Then
\begin{equation}
\label{eq:expectation_for_ACDF}
\begin{aligned}
    \mathbb{E}[(X_J + iY_J)e^{i(\theta_J+Jx)}] &= \sum_{|j|\leq d} \mathbb{E}[X_j + iY_j]e^{i(\theta_j+jx)}\Pr[J=j] \\
    &= \frac{1}{\mathcal{F}}\sum_{|j|\leq d} \Tr[\rho e^{-ij\tau H}]e^{ijx}\hat{F}_j \\
    &= \frac{\wt{C}(x)}{\mathcal{F}},
\end{aligned}
\end{equation}
where we have used \eqref{eq:measurement_outcome_expect_X} and \eqref{eq:measurement_outcome_expect_Y}. 
For simplicity we write $X_J$ and $Y_J$ into a complex random variable
\begin{equation}
\label{eq:def_Z}
Z = X_J + iY_J\in \{\pm 1 \pm i\}.    
\end{equation}
Therefore we can use 
% \begin{equation}
% \label{eq:unbiased_estimate_for_ACDF}
% G(x)=\mathcal{F}(X_J + iY_J)e^{i(\theta_J+Jx)}
% \end{equation}
\begin{equation}
\label{eq:unbiased_estimate_for_ACDF}
G(x;J,Z)=\mathcal{F}Ze^{i(\theta_J+Jx)}
\end{equation}
as an unbiased estimate of $\wt{C}(x)$.
The variance can be bounded by:
\begin{equation}
    \label{eq:variance_bound_for_the_estimate_of_ACDF}
    \begin{aligned}
    \operatorname{var}[G(x)] &\leq \mathcal{F}^2 \mathbb{E}[|X_J|^2+|Y_J|^2] 
    \leq 2\mathcal{F}^2.
    \end{aligned}
\end{equation}
Here we have used the fact that $|X_j|,|Y_j|\leq 1$.

From the above analysis, we can generate $N_s$ independent samples of $(J,Z)$, denoted by $(J_k,Z_k)$, $k=1,2,\ldots,N_s$, and then take the average 
\begin{equation}
\label{eq:G_average}
    \bar{G}(x) = \frac{1}{N_s}\sum_{k=1}^{N_s}G(x;J_k,Z_k),
\end{equation}
which can be used to estimate $\wt{C}(x)$ in an unbiased manner. The variance is upper bounded by $2\mathcal{F}^2/N_s$. In order to make the variance upper bounded by a given $\sigma^2$,  we need $N_s = \Or(\mathcal{F}^2/\sigma^2)$. The expected total evolution time is
\[
    N_s \tau \mathbb{E}[|J|] = \frac{ \mathcal{F} \tau}{\sigma^2} \sum_{|j|\leq d}|\hat{F}_j||j|.
\]
Furthermore, by Lemma~\ref{lem:approx_heaviside_function} (iii) we have $|\hat{F}_j|\leq C|j|^{-1}$ for some constant $C$. Therefore 
\[
\mathcal{F} = \Or(\log(d)), \quad \sum_{|j|\leq d}|\hat{F}_j||j| = \Or(d).
\]
The number of samples and the expected total evolution time are therefore 
\begin{equation}
    \label{eq:num_samples_and_total_evolution_time_G(x)}
    N_s = \Or\left(\frac{\log^2(d)}{\sigma^2}\right),\quad N_s \tau \mathbb{E}[|J|] = \Or\left(\frac{  \tau d\log(d)}{\sigma^2} \right),
\end{equation}
respectively. We can see that in this way we have avoided the $d^2$ dependence, which shows up in a term-by-term evaluation.

% \begin{figure}[t]
%     \centering
%     \includegraphics[width=0.45\textwidth]{ACDF.pdf}
%     \includegraphics[width=0.45\textwidth]{ACDF_zoom_in.pdf}
%     \caption{$\bar{G}(x)$ and the CDF $C(x)$, for $x\in [-\pi/3,\pi/3]$ (left) and the zoom in view around $\tau\lambda_0$ (right), the ground state energy for $\tau H$ where $H$ is the Hamiltonian for the 4-site Hubbard model with $U/t=10$ %\LL{ correct?} 
%     at half-filling. The dashed vertical line is $x=\tau\lambda_0$. The parameters are $\delta=2\times 10^{-4}$, $d=5\times 10^4$, $N_s=6000$, $\tau=\pi/(4\|H\|)$.}
%     \label{fig:ACDF}
% \end{figure}

\begin{figure}[t]
    \centering
    \includegraphics[width=0.45\textwidth]{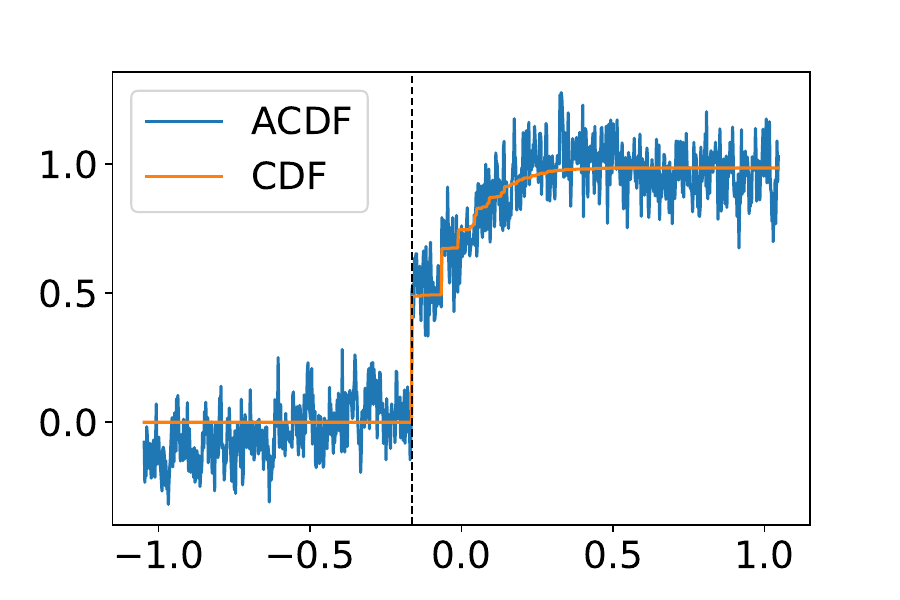}
    \includegraphics[width=0.45\textwidth]{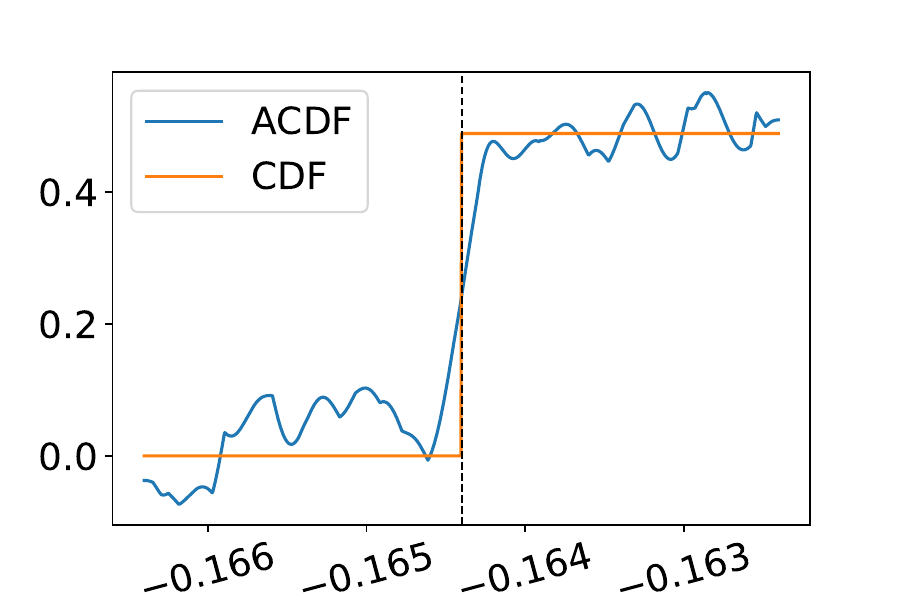}
    \caption{{$\bar{G}(x)$ and the CDF $C(x)$, for $x\in [-\pi/3,\pi/3]$ (left) and the zoom in view around $\tau\lambda_0$ (right), the ground state energy for $\tau H$ where $H$ is the Hamiltonian for the 8-site Hubbard model with $U/t=4$ %\LL{ correct?} 
    at half-filling. The dashed vertical line is $x=\tau\lambda_0$. The parameters are $\delta=2\times 10^{-4}$, $d=2\times 10^4$, $\tau=\pi/(4\|H\|)$. In total $3000$ samples are used.}}
    \label{fig:ACDF}
\end{figure}

% \begin{figure}[t]
%     \centering
%     \includegraphics[width=0.45\textwidth]{ground_energy_error.pdf}
%     \includegraphics[width=0.45\textwidth]{total_and_max_evolution_time.pdf}
%     % \includegraphics[width=0.3\textwidth]{max_evolution_time.pdf}
%     \caption{The average error (left, blue solid line), total evolution time (right, orange dashed  line), and maximal evolution time (right, blue solid line), over 10 runs, of ground state energy estimation of the 4-site Hubbard model with $U/t=4$ at half-filling. For each $\delta$, $d$ is chosen to be $d=10/\delta$, $N_s=6000$, and $\tau=\pi/(4\|H\|)$. The maximal evolution time is $\tau d=10\tau/\delta$.}
%     \label{fig:error_and_time}
% \end{figure}

\begin{figure}[t]
    \centering
    \begin{subfigure}{0.4\textwidth}
    \centering
    \includegraphics[width=0.95\linewidth]{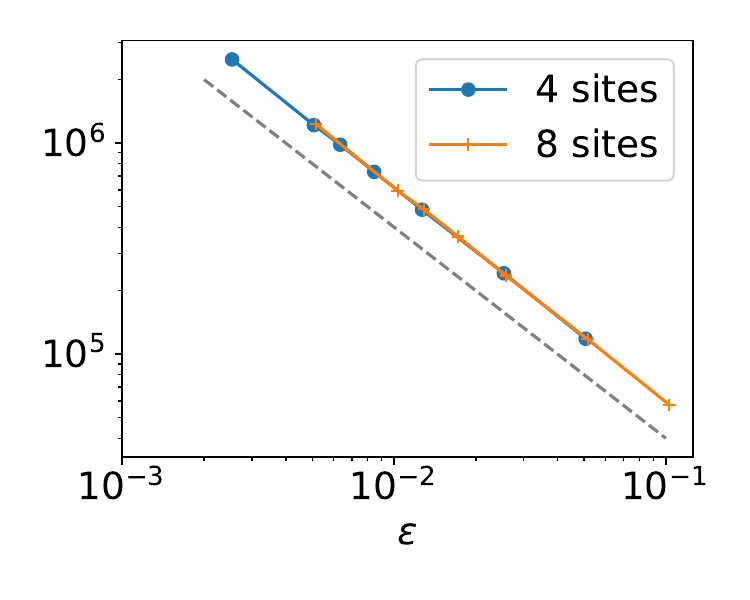}
    \caption{Total evolution time}
    \end{subfigure}
    \begin{subfigure}{0.4\textwidth}
    \centering
    \includegraphics[width=0.95\linewidth]{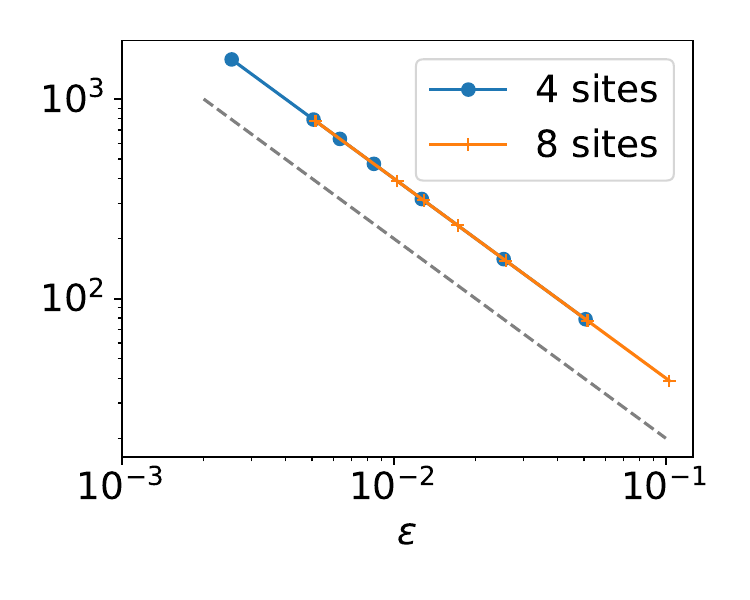}
    \caption{Maximal evolution time}
    \end{subfigure}
    \newline
    \begin{subfigure}{0.4\textwidth}
    \centering
    \includegraphics[width=0.95\linewidth]{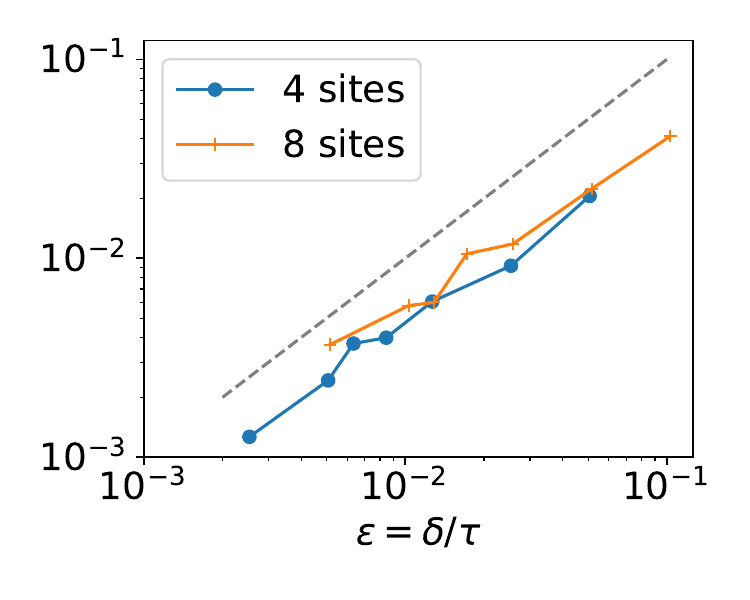}
    \caption{Average energy estimation error}
    \end{subfigure}
    \caption{{The total evolution time (a), maximal evolution time (b), and the average ground state energy estimation error (c), for 4-site and 8-site Hubbard model with $U/t=4$ at half-filling. The horizontal axis is the error threshold $\epsilon=\delta/\tau$. In (a) and (b) the grey dash lines have slope $-1$, and in (c) the grey dashed line (with slope $1$) shows the value of $\epsilon$. For each $\delta$, $d$ is chosen to be $d=4/\delta$, with $1800$ samples, and $\tau=\pi/(4\|H\|)$. The maximal evolution time is $\tau d=4\tau/\delta$.}}
    \label{fig:error_and_time}
\end{figure}

In Figure~\ref{fig:ACDF} we show the plot of the ACDF obtained through our method for the Fermi-Hubbard model. The details on this numerical experiment can be found in Appendix~\ref{sec:detail_numerical}. We can estimate the ground state energy from the ACDF in a heuristic manner: we let 
\[
x^{\star} = \inf\{x:\bar{G}(x)\geq {\eta}/2\},
\]
and $x^{\star}/\tau$ is an estimate for the ground state energy $\lambda_0$. {Here $\eta$ is chosen so that $p_0\geq \eta$.} In Section~\ref{sec:invert_CDF} we describe a more elaborate method to achieve the prescribed accuracy and confidence level. However, this heuristic method seems to work reasonably well in practice. In Figure~\ref{fig:error_and_time} we show the scaling of the ground state energy estimation error, the total evolution time, and the maximal evolution time, with respect to $\delta=\tau\epsilon$ ($\delta$ here is the parameter needed to construct $\{\hat{F}_j\}$ using Lemma~\ref{lem:approx_heaviside_function}), where $\epsilon$ is the allowed error. Both the total evolution time and the maximal evolution time are proportional to $\epsilon^{-1}$. The details on this numerical experiment can also be found in Appendix~\ref{sec:detail_numerical}.

\section{Estimating the ground state energy}
\label{sec:find_ground_energy}

In this section we discuss how to estimate the ground state energy with guaranteed error bound and confidence level from the samples generated on classical and quantum circuits discussed in Sections~\ref{sec:overview} and \ref{sec:evaluate_ACDF}. 
{First we note that the CDF $C(x)=0$ for all $-\pi/3<x<\tau\lambda_0$, and $C(x)>0$ for all $\tau\lambda_0\leq x<\pi/3$.}
{Therefore} getting the ground state energy out of the CDF can be seen as inverting the CDF{: we only need to find the smallest $x$ such that $C(x)>0$. One might consider performing a binary search to find such a point, but we run into a problem immediately: we only have access to estimates of $C(x)$ with statistical noise, and we cannot tell if the estimate is greater than zero is due to $C(x)>0$ or is merely due to statistical noise. We therefore need to make the search criterion more robust to noise.}

{Note that the CDF cannot take values between $0$ and $p_0$: $C(x)\geq p_0$ for $\tau\lambda_0\leq x<\pi/3$ and $C(x)=0$ for $-\pi/3<x<\tau\lambda_0$. Now suppose we know $p_0\geq \eta$, then for any $x$, rather than distinguishing between $C(x)=0$ and $C(x)>0$, we instead distinguish between $C(x)=0$ and $C(x)\geq \eta/2$ (here $\eta/4$ is chosen to be consistent with later discussion and it can be any number between $0$ and $1$ times $\eta$). In this setting, if the estimate of $C(x)$ is larger than $\eta/4$ then we tend to believe that $C(x)\geq \eta/2$, and if the estimate is smaller than $\eta/4$ then we tend to believe that $C(x)=0$. Thus we can tolerate an error that is smaller than $\eta/4$.}

% {We could then find the ground state energy by performing a binary search for the point at which $C(x)$ first becomes larger than $\eta/2$, assuming we could evaluate $C(x)$ up to some error below $\eta/4$. However, here we encounter another problem: we cannot evaluate $C(x)$ with precision guarantee for all $x$. The estimate available to us, namely $\bar{G}(x)$ defined in \eqref{eq:G_average}, is an unbiased estimate of the ACDF $\wt{C}(x)$ defined in \cref{eq:ACDF_definition}. Yet $\wt{C}(x)$ is a continuous function of $x$, and $C(x)$ has many jump discontinuities each of which corresponds to an eigenvalue. Therefore $\wt{C}(x)$ cannot approximate $C(x)$ uniformly. As a result we cannot perform this binary search procedure.}%Using relation \eqref{eq:relation_CDF_ACDF}, we will show that solving the following problem enables us to locate the ground state energy:}
% 

{It may appear that we can find the ground state energy by performing a binary search for the point at which $C(x)$ first becomes larger than $\eta/2$. However, we can only estimate the continuous function $\wt{C}(x)$, which cannot uniformly approximate $C(x)$. This is because $C(x)$ has many jump discontinuities (each of which corresponds to an eigenvalue). 
As a result, we cannot perform this binary search procedure directly.}

{From the above discussion we need a search criterion that can be checked via $\wt{C}(x)$.} We consider the following {criterion}:
\begin{prob}[Inverting the CDF]
\label{prob:invert_CDF} For $0<\delta<\pi/6$, $0<\eta<1$, find $x^{\star}\in(-\pi/3,\pi/3)$ such that
\begin{equation}
    \label{eq:invert_CDF_criterion}
    C(x^{\star}+\delta)>\eta/2,\quad C(x^{\star}-\delta)<\eta.
\end{equation}
\end{prob}
{Firstly we verify that this can be checked via $\wt{C}(x)$. In \eqref{eq:relation_CDF_ACDF}, if we choose $x=x^{\star}$, $\epsilon=\eta/6$, then $\wt{C}(x^{\star})>(2/3)\eta$ implies $C(x^{\star})>\eta/2$, and $\wt{C}(x^{\star})<(5/6)\eta$ implies $C(x^{\star})<\eta$. Therefore we only need to find $x^{\star}$ satisfying $(2/3)\eta<\wt{C}(x^{\star})<(5/6)\eta$ to satisfy this criterion. Secondly we show that an $x^{\star}$ satisfying this criterion gives us an estimate of the ground state energy to within additive error $\delta/\tau$.} 
% If we can solve Problem~\ref{prob:invert_CDF} for any $\eta$ and $\delta$, then we can estimate the ground state energy to within additive error $\delta/\tau$.
{Suppose we choose $\eta>0$ so that $p_0\geq \eta$. Then if we solve Problem~\ref{prob:invert_CDF} we will find an $x^{\star}$ such that $C(x^{\star}+\delta)>\eta/2 >0$ and $C(x^{\star}-\delta)<\eta\leq p_0$. $C(x^{\star}+\delta) >0$ indicates that $x^{\star}+\delta\geq\tau\lambda_0$. 
Since $C(x)$ cannot take value between $0$ and $p_0$, $C(x^{\star}-\delta)< p_0$ indicates $C(x^{\star}-\delta)=0$ and thus $x^{\star}-\delta<\tau\lambda_0$. Hence we know 
$
|x^{\star}-\tau\lambda_0|\leq \delta.
$
If we choose $\delta=\tau\epsilon$ and $\wt{\lambda}_0=x^{\star}/\tau$, then
\[
|\wt{\lambda}_0-\lambda_0|\leq \epsilon.
\]
Then $\wt{\lambda}_0$ is our desired estimate.}
% {Here we explain why solving the above Problem~\ref{prob:invert_CDF} can help us estimate the ground state energy. }

Note that \eqref{eq:invert_CDF_criterion} is a weaker requirement than $\eta/2<C(x^{\star})<\eta$, for which due to the discontinuity of $C(x)$ the required $x^{\star}$ may not exist. However an $x^{\star}$ satisfying \eqref{eq:invert_CDF_criterion} must exist. In fact, let $a=\sup\{x\in(-\pi/3,\pi/3):C(x)\leq \eta/2\}$ and $b=\inf\{x\in(-\pi/3,\pi/3):C(x)\geq \eta\}$. Then {because $C(x)$ is monotonously increasing,} $a\leq b$, 
% \LL{ reduce domain size} 
and any $x^{\star}\in[a-\delta,b+\delta)$ satisfies \eqref{eq:invert_CDF_criterion}.

% \LL{ Put this paragraph right after the (16)?} In this section we discuss how Problem~\ref{prob:invert_CDF} is related to the ground state energy estimation problem. 

Using the samples $\{J_k\}$ and $\{Z_k\}$ generated on classical and quantum circuits respectively, we are able to solve Problem~\ref{prob:invert_CDF}.
\begin{thm}[Inverting the CDF]
\label{thm:invert_CDF} 
With samples $\{J_k\}_{k=1}^{M}$ satisfying $|J_k|\leq d$ and $\{Z_k\}_{k=1}^{M}$, generated according to \eqref{eq:J_dist} and \eqref{eq:def_Z} respectively,
we can solve Problem~\ref{prob:invert_CDF} on a classical computer with probability at least $1-\vartheta$, for
$d=\Or(\delta^{-1}\log(\delta^{-1}\eta^{-1}))$ and $M=\Or(\eta^{-2}\log^2(d)(\log\log(\delta^{-1})+\log(\vartheta^{-1})))$.
% with $d=\delta^{-1}\log(\delta^{-1}\eta^{-1})$ and $N_j$ chosen according to \eqref{eq:choose_Nj_final} and $B=\Or(\log(\vartheta^{-1})+\log\log(\delta))$, 
The classical post-processing cost is 
{
\begin{equation}
    \label{eq:classical_post_processing_cost}
    % \Or\left(\eta^{-2} \log(\delta^{-1}) (\log\log(\delta^{-1})+\log(\vartheta^{-1}))(\log(\delta^{-1})+\log\log(\delta^{-1}\eta^{-1}))^{2} \right).
    \wt{\Or}(\eta^{-2}\log^{3}(\delta^{-1})\log(\vartheta^{-1})).
    % \Or\left(\delta^{-1/2}\eta^{-2}\log(\delta^{-1})(\log\log(\delta^{-1})+\log(\vartheta^{-1}))\sqrt{\log(\delta^{-1}\eta^{-1})}\right)
\end{equation}
}
To generate the samples $\{Z_k\}_{k=1}^M$ on a quantum circuit, the expected total evolution time 
and the maximal evolution time are
{
\begin{equation}
    \label{eq:total_time_time_evolution}
    \begin{aligned}
    \tau M \mathbb{E}[|J|] &= \wt{\Or}(\tau\delta^{-1}\eta^{-2}\log(\vartheta^{-1})),%\Or(\tau\delta^{-1}\eta^{-2} \log(\delta^{-1}\eta^{-1})(\log\log(\delta^{-1})+\log(\vartheta^{-1})) \\
    % &\times (\log(\delta^{-1})+\log\log(\delta^{-1}\eta^{-1}))^2).
    \end{aligned}
\end{equation}
}
% as described in \eqref{eq:total_time_time_evolution} 
and 
\begin{equation}
    \label{eq:coherence_time}
    \tau d = \Or\left(\tau\delta^{-1} \log(\delta^{-1}\eta^{-1}) \right).
\end{equation}
respectively.
\end{thm}
We will prove this theorem by constructing the algorithm for classical post-processing in Section~\ref{sec:invert_CDF}.
Since solving Problem~\ref{prob:invert_CDF} enables us to estimate the ground state energy as discussed above, from Theorem~\ref{thm:invert_CDF} we have the following corollary:
\begin{cor}[Ground state energy]
\label{cor:ground_energy}
With samples $\{J_k\}_{k=1}^{M}$ satisfying $|J_k|\leq d$ and $\{Z_k\}_{k=1}^{M}$, generated according to \eqref{eq:J_dist} and \eqref{eq:def_Z} respectively,
we can estimate the ground state energy $\lambda_0$ to within additive error $\epsilon$ on a classical computer with probability at least $1-\vartheta$,
if {$p_0\geq \eta$} for some known $\eta$, 
$d=\Or(\epsilon^{-1}\tau^{-1}\log(\epsilon^{-1}\tau^{-1}\eta^{-1}))$, and $M=\Or(\eta^{-2}\log^2(d)(\log\log(\epsilon^{-1}\tau^{-1})+\log(\vartheta^{-1})))$.
% with $d=\delta^{-1}\log(\delta^{-1}\eta^{-1})$ and $N_j$ chosen according to \eqref{eq:choose_Nj_final} and $B=\Or(\log(\vartheta^{-1})+\log\log(\delta))$, 
The classical post-processing cost is $\Or(\eta^{-2}\polylog(\epsilon^{-1}\tau^{-1}\eta^{-1}))$. The expected total evolution time and the maximal evolution time are $\Or(\epsilon^{-1}\eta^{-2}\polylog(\epsilon^{-1}\tau^{-1}\eta^{-1}))$ and $\Or(\epsilon^{-1}\polylog(\epsilon^{-1}\tau^{-1}\eta^{-1}))$ respectively.
\end{cor}

Usually the Heisenberg limit is defined in terms of the root-mean-square error (RMSE) of the estimate. In this paper we focus on ensuring the error of the ground state energy to be below a threshold $\epsilon$ with probability at least $1-\vartheta$. From Corollary~\ref{cor:ground_energy}, our algorithm only has a logarithmic dependence on $\vartheta^{-1}$, and the error can be at most $2\|H\|$, we can easily ensure the RMSE is $\Or(\epsilon)$ using the result by choosing $\vartheta=\Or(\epsilon^2\|H\|^{-2})$. We can see the total evolution time scaling with respect to $\epsilon$ is still $\wt{\Or}(\epsilon^{-1})$.

\begin{rem}[System size dependence]
%\YT{Probably we can delete this paragraph altogether.} If we do not know a lower bound for $p_0$, {here we introduce a method to solve the problem heuristically}. We choose $\eta = 2^{-\ell}$ for $\ell=1,2,3,\ldots$. For each $\eta$ we solve Problem~\ref{prob:invert_CDF} with $\delta=\tau\epsilon$, then if no error occurs, we will end up getting a sequence of outputs $x^{\star}_{\ell}$ each of which is greater than $\tau(\lambda_0-\epsilon)$. Eventually for large enough $\ell=\Or(\log(p_0^{-1}))$ we will have $|x^{\star}_{\ell}/\tau-\lambda_0|\leq \epsilon$. Although we cannot know it for sure by looking at $x^{\star}_{\ell}$, in practice we will see many $x^{\star}_{\ell}$ concentrated in the interval $[\tau(\lambda_0-\epsilon),\tau(\lambda_0+\epsilon)]$, and we can pick an estimate among them. This can be further complemented by some prior knowledge about the ground state energy.

{One might notice the absence of an explicit system size dependence in the evolution time scaling in Theorem~\ref{thm:invert_CDF} and Corollary~\ref{cor:ground_energy}. This is because, as mentioned before in the Introduction, the total evolution time depends on the system size indirectly through two parameters $\tau$ and $\eta$. Moreover, if we consider the dependence of the total runtime on the system size, we also need to account for the overhead that comes from performing Hamiltonian simulation. This overhead and the scaling of $\eta$ with respect to the system size are highly problem-specific and are independent from the tasks we are considering in this paper, and hence we will not discuss them in more detail. Because the Hamiltonian norm can generally be upper bounded by a polynomial of the system size, and the total evolution time dependence on $\tau^{-1}$ is poly-logarithmic, $\tau$ contributes a poly-logarithmic overhead in the system size dependence.} %\LL{ rephrase the "tangential" part? Also add the sentence of the system size dependence of $\tau$ here (as mentioned in the letter)?} \YT{Changed "tangential" to "independent from". Added the $\tau$ scaling}
\end{rem}

\section{Inverting the CDF}
\label{sec:invert_CDF}

In this section we prove Theorem~\ref{thm:invert_CDF} by constructing the classical post-processing algorithm to solve Problem~\ref{prob:invert_CDF}
 {using samples from a quantum circuit}.
{Since we want to search for an $x^{\star}$ satisfying the requirement \eqref{eq:invert_CDF_criterion}, a natural idea is to use binary search. Our setting is somewhat different from the usual binary search setting, but we will show that a similar approach still works. The current setting differs from the setting of binary search mainly in two ways: first any $x^{\star}\in[\tau\lambda_0-\delta,\tau\lambda_0+\delta]$ satisfies the requirement \eqref{eq:invert_CDF_criterion} and can therefore be a target. When performing binary search we want to be able to tell if the target is to the left or right of a given $x$, but here the targets may be on both sides of $x$. When this happens there is some uncertainty as to how the algorithm will proceed next. However in our algorithm we will show that this does not present a problem. Also, because this algorithm is based on random samples, there is some failure probability in each search step. We will use a majority voting procedure to suppress the failure probability so that in the end the algorithm will produce a correct answer with probability arbitrarily close to $1$.}

{We suppose we are} given independent samples of $(J,Z)$ defined in \eqref{eq:J_dist} and \eqref{eq:def_Z} {generated from a quantum circuit}.
% $\{J_k\}_{k=1}^M$ and $\{Z_k\}_{k=1}^M$. 
We denote these samples by
$\{(J_k,Z_k)\}_{k=1}^M$. 
We divide them into $N_b$ batches of size $N_s$, where $N_s N_b=M$. {This division is for the majority voting procedure we mentioned above.} The maximal evolution time needed to generate these samples is proportional to $\max_k |J_k|\leq d$. The expected total evolution time we will need is proportional to $M\mathbb{E}[|J|]$.

We first reduce Problem~\ref{prob:invert_CDF} into a decision problem. For any $x\in (-\pi/3,\pi/3)$, %\LL{ reduce domain size} 
one of the following must be true:
\begin{equation}
\label{eq:certify_one_of_them}
C(x+\delta)>\eta/2,\quad \text{or} \quad C(x-\delta)<\eta.
\end{equation}
If there is a subroutine that tells us which one of the two is correct, or randomly picks one when both are correct, then we can use it to find $x^{\star}$. We assume such a subroutine, which uses $\{(J_k,Z_k)\}_{k=1}^M$, exists and denote it by the name $\texttt{CERTIFY}(x,\delta,\eta, \{(J_k,Z_k)\})$. The subroutine returns either $0$ or $1$: $0$ for $C(x+\delta)>\eta/2$ being true, and $1$ for $C(x-\delta)<\eta$ being true.

In Algorithm~\ref{alg:invert_CDF}, with $\texttt{CERTIFY}(x,\delta,\eta, \{(J_k,Z_k)\})$, we describe the algorithm to solve Problem~\ref{prob:invert_CDF}. This algorithm we denote by $\texttt{INVERT\_CDF}(\delta,\eta, \{(J_k,Z_k)\})$.  It runs as follows: we start with $x_{0,0}=-\pi/3$ and $x_{1,0}=\pi/3$. They are chosen so that $C(x_{1,0})>\eta/2$ and $C(x_{0,0})<\eta$. Let $\ell$ be the number of iterations we have performed, and $\ell=0$ at the beginning. At each iteration, we let $x_{\ell}=(x_{0,\ell}+x_{1,\ell})/2$, and run $\texttt{CERTIFY}(x_{\ell},(2/3)\delta,\eta, \{(J_k,Z_k)\})$. This tells us either $C(x_{\ell}+(2/3)\delta)>\eta/2$ or $C(x_{\ell}-(2/3)\delta)<\eta$. If the former then we let $x_{0,\ell+1}=x_{0,\ell}$, $x_{1,\ell+1}=x_{\ell}+(2/3)\delta$, and if the latter we let $x_{0,\ell+1}=x_{\ell}+(2/3)\delta$, $x_{1,\ell+1}=x_{1,\ell}$. This is done so that for each $\ell$ we have
\begin{equation}
\label{eq:invert_CDF_iter_always_true}
C(x_{0,\ell})<\eta,\quad C(x_{1,\ell})>\eta/2.
\end{equation}
We then let $\ell\leftarrow \ell+1$ and go to the next iteration. The algorithm stops once $x_{1,\ell}-x_{0,\ell}\leq 2\delta$. We denote the total number of iterations by $L$. The output is $x_L = (x_{0,L}+x_{1,L})/2$. Because \eqref{eq:invert_CDF_iter_always_true} holds for each iteration we have
\[
C(x_L-\delta) \leq C(x_{0,L}) <\eta,\quad C(x_L+\delta) \geq C(x_{1,L}) >\eta/2.
\]
Thus we can see $x_L$ satisfies the requirements for $x^{\star}$ in Problem~\ref{prob:invert_CDF}. 
The next question is, how many iterations does it take to satisfy the stopping criterion? Regardless of the outcome of the $\texttt{CERTIFY}$ subroutine, we always have
\[
x_{1,\ell+1}-x_{0,\ell+1}=\frac{1}{2}(x_{1,\ell}-x_{0,\ell}) + \frac{2}{3}\delta.
\]
From this we can see
\[
x_{1,\ell}-x_{0,\ell} = \frac{2\pi/3-(4/3)\delta}{2^{\ell}}+\frac{4}{3}\delta.
\]
Therefore it takes $L=\Or(\log(\delta^{-1}))$ iterations for the algorithm to stop.

\begin{algorithm}[ht]
\caption{$\texttt{INVERT\_CDF}$}
\label{alg:invert_CDF}
\begin{algorithmic}
\REQUIRE $\delta,\eta, \{(J_k,Z_k)\}$ 
\STATE $x_0\leftarrow -\pi/3$, $x_1\leftarrow \pi/3$;
\WHILE{$x_1-x_0> 2\delta$}
\STATE $x \leftarrow (x_0+x_1)/2$;
% \STATE $\delta\leftarrow (x_1-x_0)/6$;
\STATE $u\leftarrow \texttt{CERTIFY}(x,(2/3)\delta,\eta, \{(J_k,Z_k)\})$;
\IF{$u=0$}
\STATE $x_1\leftarrow x+(2/3)\delta$;
\ELSE
\STATE $x_0\leftarrow x-(2/3)\delta$;
\ENDIF
\ENDWHILE
\ENSURE $(x_0+x_1)/2$
\end{algorithmic}
\end{algorithm}

Next we discuss how to construct the subroutine $\texttt{CERTIFY}(x,\delta,\eta,\{(J_k,Z_k)\})$. While we cannot directly evaluate the CDF $C(x)$ for any $x$, we can estimate the ACDF $\wt{C}(x)$ using the data $\{J_k\}$ and $\{Z_k\}$. We can let $\epsilon=\eta/8$ in \eqref{eq:heaviside_approx_criterion_main_text} and choose $d=\Or(\delta^{-1}\log(\delta^{-1}\eta^{-1}))$ according to Lemma~\ref{lem:approx_heaviside_function}. Then by \eqref{eq:relation_CDF_ACDF}, we have $C(x-\delta)\leq \wt{C}(x)+\eta/8$ and $C(x+\delta)\geq \wt{C}(x)-\eta/8$. One of the following must be true:
\begin{equation}
    \label{eq:ACDF_always_true}
    \wt{C}(x)>(5/8)\eta,\quad\text{or}\quad \wt{C}(x)<(7/8)\eta,
\end{equation}
then the former implies $C(x+\delta)>\eta/2$ and the latter $C(x-\delta)<\eta$. Therefore the $\texttt{CERTIFY}$ subroutine only needs to decide which one of the two is correct or to output a random choice when both are correct.

As discussed in Section~\ref{sec:evaluate_ACDF}, %$\wt{C}_r(x)$, with $r=1,2,\ldots,B$ indexing batches, are unbiased estimates of $\wt{C}(x)$. 
$\bar{G}(x)$ is an unbiased estimate of $\wt{C}(x)$. We use $\{J_k\}$ and $\{Z_k\}$ to get $N_b$ samples for $\bar{G}(x)$, denoted by $\bar{G}_r(x)$, via
\[
\bar{G}_r(x)= \frac{1}{N_s}\sum_{k=1}^{N_s}G(x;J_{(r-1)N_s+k},Z_{(r-1)N_s+k})
\]
for $r=1,2,\ldots,N_b$. Here $G(x;J,Z)$ is defined in \eqref{eq:unbiased_estimate_for_ACDF}.
For each $r$, we compare $\bar{G}_r(x)$ with $(3/4)\eta$. If $\bar{G}_r(x)>(3/4)\eta$ for a majority of batches, then we tend to believe $\wt{C}(x)>(5/8)\eta$ and output $0$ for $C(x+\delta)>\eta/2$. Otherwise, we tend to believe $\wt{C}(x)<(7/8)\eta$ and output $1$ for 
% $C(x+\delta)>\eta/2$\LL{ $C(x-\delta)<\eta$?} 
$C(x-\delta)<\eta$. {This is the majority voting procedure we mentioned earlier.} For the pseudocode for the subroutine see Algorithm~\ref{alg:certify}.

\begin{algorithm}[ht]
\caption{$\texttt{CERTIFY}$}
\label{alg:certify}
\begin{algorithmic}
\REQUIRE $x,\delta,\eta, \{(J_k,Z_k)\}$
\STATE $b\leftarrow 0$, $c\leftarrow 0$;
\FOR{$r=1,2,\ldots,N_b$}
% \STATE Generate samples $G_k(x)$, $k=1,2,\ldots,N_s$, according to \eqref{eq:unbiased_estimate_for_ACDF};
\STATE $\bar{G}_r(x)\leftarrow (1/N_s)\sum_{k=1}^{N_s}G(x;J_{(r-1)N_s+k},Z_{(r-1)N_s+k})$; \COMMENT{$G(x;J,Z)$ defined in \eqref{eq:unbiased_estimate_for_ACDF}}
% \STATE Evaluate $\wt{C}_r(x)$ for $r=1,2,\ldots,B$ by \eqref{eq:estimate_ACDF}; $c\leftarrow 0$;
% \FOR{$r=1,2,\ldots,B$}
\IF{$\bar{G}_r(x)>(3/4)\eta$}
\STATE $c\leftarrow c+1$;
\ENDIF
% \ENDFOR
% \STATE $c\leftarrow |\{r\in\{1,2,\ldots,B\}:\wt{C}_r(x)>\eta/4\LL{ $3\eta/4$?} \}|$;  $b\leftarrow 0$;
% \STATE $c\leftarrow |\{r\in\{1,2,\ldots,B\}:\wt{C}_r(x)>(3/4)\eta \}|$;  $b\leftarrow 0$;
\ENDFOR
\IF{$c\leq B/2$}
\STATE $b\leftarrow 1$;
\ENDIF
\ENSURE $b$
\end{algorithmic}
\end{algorithm}

In the $\texttt{CERTIFY}$ subroutine, an error occurs when $\wt{C}(x)>(5/8)\eta$ yet a majority of estimates $\bar{G}_r(x)$ are smaller than $(3/4)\eta$, or when $\wt{C}(x)<(7/8)\eta$ yet a majority of estimates $\bar{G}_r(x)$ are larger than $(3/4)\eta$. We need to make the probability of this kind of error occurring upper bounded by $\nu$. First we assume $\wt{C}(x)>(5/8)\eta$. Then for each $r$, by Markov's inequality, we have
\[
\Pr[\bar{G}_r(x)<(3/4)\eta] \leq \frac{64\operatorname{var}[\bar{G}_r(x)]}{\eta^2}.
\]
We want to make this probability at most $1/4$. Therefore we need $\operatorname{var}[\bar{G}_r(x)]\leq \eta^2/256$. To ensure this, by \eqref{eq:num_samples_and_total_evolution_time_G(x)} in which we let $\sigma^2 = \eta^2/256$, we can choose
\begin{equation}
\label{eq:choose_Ns_final}
N_s = \Or\left(\frac{\log^2(d)}{\eta^2}\right).
\end{equation}
Then by the Chernoff bound the probability of the majority of estimates $\bar{G}_r(x)$ being smaller than $(3/4)\eta$ is at most $e^{-C' N_b}$ for some constant $C'$. In order to make this probability bounded by $\nu$ we only need to let $N_b=\Or(\log(\nu^{-1}))$.

In the algorithm $\texttt{INVERT\_CDF}$, the subroutine $\texttt{CERTIFY}$ is used $L=\Or(\log(\delta^{-1}))$ times. If an error occurs in a single run of $\texttt{CERTIFY}$ with probability at most $\nu$ then in the total $L$ times we use this subroutine the probability of an error occurring is at most $L\nu$. Therefore in order to ensure that an error occurs with probability at most $\vartheta$ in $\texttt{INVERT\_CDF}$, we need to set $\nu=\vartheta/L$. Therefore $N_b=\Or(\log(L\vartheta^{-1}))=\Or(\log\log(\delta^{-1})+\log(\vartheta^{-1}))$.

The above analysis shows that in order to solve Problem~\ref{prob:invert_CDF} the total evolution time is $M\mathbb{E}[|J|]=N_b N_s\mathbb{E}[|J|]$. We evaluate $N_s\mathbb{E}[|J|]$ by \eqref{eq:num_samples_and_total_evolution_time_G(x)} in which we let $\sigma^2=\eta^2/256$ as discussed before when we estimate how large $N_s$ needs to be in \eqref{eq:choose_Ns_final}. Multiplying this by $N_b$ we have \eqref{eq:total_time_time_evolution}. Note here we do not need to multiply by $L$ because in each $\texttt{CERTIFY}$ subroutine we can reuse the same $\{J_k\}$, $\{Z_k\}$.
% \begin{equation}
%     \label{eq:total_time_time_evolution}
%     \tau N_s N_b \mathbb{E}[|J|] = \Or\left(\tau\delta^{-1}\eta^{-2} \log(\delta^{-1}\eta^{-1})(\log\log(\delta^{-1})+\log(\vartheta^{-1}))(\log(\delta^{-1})+\log\log(\delta^{-1}\eta^{-1}))^2\right).
% \end{equation}
The maximal evolution time required is $\tau d$ and this leads to \eqref{eq:coherence_time}.
% \begin{equation}
%     \label{eq:coherence_time}
%     \tau d = \Or\left(\tau\delta^{-1} \log(\delta^{-1}\eta^{-1}) \right).
% \end{equation}
The main cost in classical post-processing comes from evaluating $\bar{G}_r(x)$. This needs to be done $LN_b$ times. Each evaluation involves $\Or(N_s)=\Or(\eta^{-2}\log^2(d))$ arithmetic operations. The total runtime for classical post-processing is therefore $LN_b N_s=LM$, which leads to \eqref{eq:classical_post_processing_cost}. Thus we have obtained all the cost estimates in Theorem~\ref{thm:invert_CDF} and proved the theorem.

% The complexity of our algorithm can be understood in a more intuitive way (we omit all poly-logarithmic factors for clarity): we need to estimate the CDF to precision $\Or(\eta)$ through Monte-Carlo sampling, and this requires $\Or(\eta^{-2})$ samples. Generating each sample requires running a time evolution for time $\Or(\epsilon^{-1})$. Therefore the maximal evolution time is $\Or(\epsilon^{-1})$ and the total evolution time is $\Or(\epsilon^{-2}\eta^{-2})$.

\section{Discussions}

% Summarize the scaling of the algorithm. Bridging the gap between $p_0^{-2}$ and $p_0^{-1/2}$ scaling in the total runtime is unclear. Trotter error might affect the analysis and scaling with respect to $\epsilon$.
% Comment that other than MLMC, similar scaling can be achieved by deterministic optimization of the number of samples for each $j$.

In this paper we presented an algorithm to estimate the ground state energy with Heisenberg-limited precision scaling. The quantum circuit we used requires only one ancilla qubit, and the maximal evolution time needed {per run} has a poly-logarithmic dependence on the overlap $p_0$. Such dependence on $p_0$ is exponentially better than that required by QPE using a similarly structured circuit {using semi-classical Fourier transform, as discussed in Section~\ref{sec:related}}. Both rigorous analysis and numerical experiments are done to validate the correctness and efficiency of our algorithm. 

Although our algorithm has a near-optimal dependence on the precision, the dependence on $p_0$ (more precisely, on its lower bound $\eta$), which scales as $p_0^{-2}$ in Corollary~\ref{cor:ground_energy}, is far from optimal compared to the $p_0^{-1/2}$ scaling in Refs.~\cite{ge2019faster,lin2020near}. Whether one can achieve this $p_0^{-1/2}$ scaling without using a quantum circuit with substantially larger maximal evolution time, and without using such techniques as LCU or block-encoding, remains an open question.

The probabilistic choice of the simulation time according to \cref{eq:J_dist} plays an important role in reducing the total evolution time. However, we may partially derandomize the algorithm following the spirit of the multilevel Monte Carlo (MLMC) method~\cite{Giles2015} in the classical setting. The method we developed for computing the approximate CDF in Section~\ref{sec:evaluate_ACDF} is in fact a quite general approach for evaluating expectation values from matrix functions. This method can act as a substitute of the LCU method in many cases, especially in a near-term setting. Using this method to compute other properties of the spectrum, such as the spectral density, is a direction for future work.

\section*{Acknowledgments}
This work was partially supported by the Air Force Office of Scientific
Research under award number FA9550-18-1-0095 (L.L. and Y.T.), and by the
Department of Energy under Grant No. DE-SC0017867 and under the Quantum
Systems Accelerator program (L.L.). We thank Andrew Baczewski and Barbara Terhal for helpful discussions.

\bibliographystyle{abbrvnat}
\bibliography{ref}

\begin{thebibliography}{73}
\providecommand{\natexlab}[1]{#1}
\providecommand{\url}[1]{\texttt{#1}}
\expandafter\ifx\csname urlstyle\endcsname\relax
  \providecommand{\doi}[1]{doi: #1}\else
  \providecommand{\doi}{doi: \begingroup \urlstyle{rm}\Url}\fi

\bibitem[Abrams and Lloyd(1999)]{AbramsLloyd1999quantum}
D.~S. Abrams and S.~Lloyd.
\newblock Quantum algorithm providing exponential speed increase for finding
  eigenvalues and eigenvectors.
\newblock \emph{Phys. Rev. Lett.}, 83\penalty0 (24):\penalty0 5162, 1999.
\newblock \doi{10.1103/PhysRevLett.83.5162}.

\bibitem[Aharonov et~al.(2009)Aharonov, Gottesman, Irani, and
  Kempe]{AharonovGottesmanEtAl2009}
D.~Aharonov, D.~Gottesman, S.~Irani, and J.~Kempe.
\newblock The power of quantum systems on a line.
\newblock \emph{Comm. Math. Phys.}, 287\penalty0 (1):\penalty0 41--65, 2009.
\newblock \doi{10.1007/s00220-008-0710-3}.

\bibitem[Aharonov and Bohm(1961)]{AharonovBohm1961time}
Y.~Aharonov and D.~Bohm.
\newblock Time in the quantum theory and the uncertainty relation for time and
  energy.
\newblock \emph{Phys. Rev.}, 122\penalty0 (5):\penalty0 1649, 1961.

\bibitem[Aharonov et~al.(2002)Aharonov, Massar, and
  Popescu]{AharonovMassarPopescu2002measuring}
Y.~Aharonov, S.~Massar, and S.~Popescu.
\newblock Measuring energy, estimating {Hamiltonians}, and the time-energy
  uncertainty relation.
\newblock \emph{Phys. Rev. A}, 66\penalty0 (5):\penalty0 052107, 2002.
\newblock \doi{10.1103/PhysRevA.66.052107}.

\bibitem[Atia and Aharonov(2017)]{AtiaAharonov2017}
Y.~Atia and D.~Aharonov.
\newblock Fast-forwarding of {Hamiltonians} and exponentially precise
  measurements.
\newblock \emph{Nature Comm.}, 8\penalty0 (1), 2017.
\newblock \doi{10.1038/s41467-017-01637-7}.

\bibitem[Babbush et~al.(2015)Babbush, McClean, Wecker, Aspuru-Guzik, and
  Wiebe]{BabbushMcCleanEtAl2015chemical}
R.~Babbush, J.~McClean, D.~Wecker, A.~Aspuru-Guzik, and N.~Wiebe.
\newblock Chemical basis of trotter-suzuki errors in quantum chemistry
  simulation.
\newblock \emph{Physical Review A}, 91\penalty0 (2):\penalty0 022311, 2015.

\bibitem[Babbush et~al.(2018)Babbush, Gidney, Berry, Wiebe, McClean, Paler,
  Fowler, and Neven]{BabbushGidneyEtAl2018encoding}
R.~Babbush, C.~Gidney, D.~W. Berry, N.~Wiebe, J.~McClean, A.~Paler, A.~Fowler,
  and H.~Neven.
\newblock Encoding electronic spectra in quantum circuits with linear {T}
  complexity.
\newblock \emph{Physical Review X}, 8\penalty0 (4):\penalty0 041015, 2018.

\bibitem[Babbush et~al.(2021)Babbush, McClean, Newman, Gidney, Boixo, and
  Neven]{BabbushMcCleanEtAl2021focus}
R.~Babbush, J.~R. McClean, M.~Newman, C.~Gidney, S.~Boixo, and H.~Neven.
\newblock Focus beyond quadratic speedups for error-corrected quantum
  advantage.
\newblock \emph{PRX Quantum}, 2\penalty0 (1):\penalty0 010103, 2021.

\bibitem[Berry et~al.(2009)Berry, Higgins, Bartlett, Mitchell, Pryde, and
  Wiseman]{BerryHiggins2009}
D.~W. Berry, B.~L. Higgins, S.~D. Bartlett, M.~W. Mitchell, G.~J. Pryde, and
  H.~M. Wiseman.
\newblock How to perform the most accurate possible phase measurements.
\newblock \emph{Phys. Rev. A}, 80\penalty0 (5), 2009.
\newblock \doi{10.1103/physreva.80.052114}.

\bibitem[Berry et~al.(2015)Berry, Childs, and Kothari]{BerryChildsKothari2015}
D.~W. Berry, A.~M. Childs, and R.~Kothari.
\newblock {Hamiltonian} simulation with nearly optimal dependence on all
  parameters.
\newblock In \emph{2015 IEEE 56th Annual Symposium on Foundations of Computer
  Science}, pages 792--809. IEEE, 2015.

\bibitem[Berry et~al.(2020)Berry, Childs, Su, Wang, and
  Wiebe]{BerryChildsSuWangWiebe2020}
D.~W. Berry, A.~M. Childs, Y.~Su, X.~Wang, and N.~Wiebe.
\newblock Time-dependent hamiltonian simulation with l1-norm scaling.
\newblock \emph{Quantum}, 4:\penalty0 254, 2020.
\newblock \doi{10.22331/q-2020-04-20-254}.

\bibitem[Bittel and Kliesch(2021)]{bittel2021VQE_NPhard}
C.~Bittel and M.~Kliesch.
\newblock Training variational quantum algorithms is {NP}-hard -- even for
  logarithmically many qubits and free fermionic systems.
\newblock \emph{arXiv preprint arXiv:2101.07267}, 2021.

\bibitem[Boixo and Somma(2008)]{BoixoSomma2008parameter}
S.~Boixo and R.~D. Somma.
\newblock Parameter estimation with mixed-state quantum computation.
\newblock \emph{Physical Review A}, 77\penalty0 (5):\penalty0 052320, 2008.

\bibitem[Booth et~al.(2021)Booth, O'Gorman, Marshall, Hadfield, and
  Rieffel]{BoothOGorman2021quantum}
K.~E. Booth, B.~O'Gorman, J.~Marshall, S.~Hadfield, and E.~Rieffel.
\newblock Quantum-accelerated constraint programming.
\newblock \emph{arXiv preprint arXiv:2103.04502}, 2021.

\bibitem[Campbell(2019)]{Campbell2019random}
E.~Campbell.
\newblock Random compiler for fast hamiltonian simulation.
\newblock \emph{Phys. Rev. Lett.}, 123\penalty0 (7), 2019.
\newblock \doi{10.1103/physrevlett.123.070503}.

\bibitem[Campbell(2021)]{campbell2020early}
E.~T. Campbell.
\newblock Early fault-tolerant simulations of the {Hubbard} model.
\newblock \emph{Quantum Science and Technology}, 7\penalty0 (1):\penalty0
  015007, 2021.

\bibitem[Chakraborty et~al.(2018)Chakraborty, Gily{\'e}n, and
  Jeffery]{chakraborty2018power}
S.~Chakraborty, A.~Gily{\'e}n, and S.~Jeffery.
\newblock The power of block-encoded matrix powers: improved regression
  techniques via faster hamiltonian simulation.
\newblock \emph{arXiv preprint arXiv:1804.01973}, 2018.

\bibitem[Chen et~al.(2020)Chen, Huang, Kueng, and
  Tropp]{ChenHuangKuengTropp2020quantum}
C.-F. Chen, H.-Y. Huang, R.~Kueng, and J.~A. Tropp.
\newblock Quantum simulation via randomized product formulas: Low gate
  complexity with accuracy guarantees.
\newblock \emph{arXiv preprint arXiv:2008.11751}, 2020.

\bibitem[Childs and Su(2019)]{ChildsSu2019nearly}
A.~M. Childs and Y.~Su.
\newblock Nearly optimal lattice simulation by product formulas.
\newblock \emph{Physical review letters}, 123\penalty0 (5):\penalty0 050503,
  2019.

\bibitem[Childs et~al.(2000)Childs, Preskill, and
  Renes]{ChildsPreskillRenes2000quantum}
A.~M. Childs, J.~Preskill, and J.~Renes.
\newblock Quantum information and precision measurement.
\newblock \emph{J. Mod. Optics}, 47\penalty0 (2-3):\penalty0 155--176, 2000.

\bibitem[Childs et~al.(2021)Childs, Su, Tran, Wiebe, and
  Zhu]{ChildsSuTranWiebeZhu2021}
A.~M. Childs, Y.~Su, M.~C. Tran, N.~Wiebe, and S.~Zhu.
\newblock Theory of {Trotter} error with commutator scaling.
\newblock \emph{Phys. Rev. X}, 11\penalty0 (1), 2021.
\newblock \doi{10.1103/physrevx.11.011020}.

\bibitem[Cleve et~al.(1998)Cleve, Ekert, Macchiavello, and
  Mosca]{CleveEkertEtAl1998quantum}
R.~Cleve, A.~Ekert, C.~Macchiavello, and M.~Mosca.
\newblock Quantum algorithms revisited.
\newblock \emph{Proceedings of the Royal Society of London. Series A:
  Mathematical, Physical and Engineering Sciences}, 454\penalty0
  (1969):\penalty0 339--354, 1998.

\bibitem[Emerson et~al.(2004)Emerson, Lloyd, Poulin, and
  Cory]{EmersonLloydEtAl2004estimation}
J.~Emerson, S.~Lloyd, D.~Poulin, and D.~Cory.
\newblock Estimation of the local density of states on a quantum computer.
\newblock \emph{Physical Review A}, 69\penalty0 (5):\penalty0 050305, 2004.

\bibitem[Ge et~al.(2019)Ge, Tura, and Cirac]{ge2019faster}
Y.~Ge, J.~Tura, and J.~I. Cirac.
\newblock Faster ground state preparation and high-precision ground energy
  estimation with fewer qubits.
\newblock \emph{J. Math. Phys.}, 60\penalty0 (2):\penalty0 022202, 2019.
\newblock \doi{10.1063/1.5027484}.

\bibitem[Giles(2015)]{Giles2015}
M.~B. Giles.
\newblock Multilevel monte carlo methods.
\newblock \emph{Acta Numer.}, 24:\penalty0 259--328, 2015.

\bibitem[Gily{\'e}n et~al.(2019)Gily{\'e}n, Su, Low, and
  Wiebe]{gilyen2019quantum}
A.~Gily{\'e}n, Y.~Su, G.~H. Low, and N.~Wiebe.
\newblock Quantum singular value transformation and beyond: exponential
  improvements for quantum matrix arithmetics.
\newblock In \emph{Proceedings of the 51st Annual ACM SIGACT Symposium on
  Theory of Computing}, pages 193--204, 2019.
\newblock \doi{10.1145/3313276.3316366}.

\bibitem[Giovannetti et~al.(2006)Giovannetti, Lloyd, and
  Maccone]{GiovannettiLloydMaccone2006}
V.~Giovannetti, S.~Lloyd, and L.~Maccone.
\newblock Quantum metrology.
\newblock \emph{Phys. Rev. Lett.}, 96\penalty0 (1), 2006.
\newblock \doi{10.1103/physrevlett.96.010401}.

\bibitem[Giovannetti et~al.(2011)Giovannetti, Lloyd, and
  Maccone]{GiovannettiLloydMaccone2011advances}
V.~Giovannetti, S.~Lloyd, and L.~Maccone.
\newblock Advances in quantum metrology.
\newblock \emph{Nature Photon.}, 5\penalty0 (4):\penalty0 222, 2011.
\newblock \doi{10.1038/nphoton.2011.35}.

\bibitem[Griffiths and Niu(1996)]{griffiths1996semiclassical}
R.~B. Griffiths and C.-S. Niu.
\newblock Semiclassical fourier transform for quantum computation.
\newblock \emph{Phys. Rev. Lett.}, 76\penalty0 (17):\penalty0 3228, 1996.
\newblock \doi{10.1103/physrevlett.76.3228}.

\bibitem[Higgins et~al.(2007)Higgins, Berry, Bartlett, Wiseman, and
  Pryde]{HigginsBerryEtAl2007}
B.~L. Higgins, D.~W. Berry, S.~D. Bartlett, H.~M. Wiseman, and G.~J. Pryde.
\newblock Entanglement-free heisenberg-limited phase estimation.
\newblock \emph{Nature}, 450\penalty0 (7168):\penalty0 393--396, 2007.
\newblock \doi{10.1038/nature06257}.

\bibitem[Huggins et~al.(2020)Huggins, Lee, Baek, O'Gorman, and
  Whaley]{Huggins2020nonorthogonalVQE}
W.~J. Huggins, J.~Lee, U.~Baek, B.~O'Gorman, and K.~B. Whaley.
\newblock A non-orthogonal variational quantum eigensolver.
\newblock \emph{New J. of Phys.}, 22\penalty0 (7):\penalty0 073009, 2020.
\newblock \doi{10.1088/1367-2630/ab867b}.

\bibitem[Kempe et~al.(2006)Kempe, Kitaev, and Regev]{KempeKitaevRegev2006}
J.~Kempe, A.~Kitaev, and O.~Regev.
\newblock The complexity of the local {Hamiltonian} problem.
\newblock \emph{SIAM J. Comput.}, 35\penalty0 (5):\penalty0 1070--1097, 2006.
\newblock \doi{10.1007/978-3-540-30538-5_31}.

\bibitem[Kitaev(1995)]{Kitaev1995}
A.~Y. Kitaev.
\newblock Quantum measurements and the abelian stabilizer problem.
\newblock \emph{arXiv preprint quant-ph/9511026}, 1995.

\bibitem[Kitaev et~al.(2002)Kitaev, Shen, and Vyalyi]{KitaevShenVyalyi2002}
A.~Y. Kitaev, A.~Shen, and M.~N. Vyalyi.
\newblock \emph{Classical and quantum computation}.
\newblock Number~47 in Graduate Studies in Mathematics. American Mathematical
  Soc., 2002.
\newblock \doi{10.1090/gsm/047}.

\bibitem[Kivlichan et~al.(2018)Kivlichan, McClean, Wiebe, Gidney, Aspuru-Guzik,
  Chan, and Babbush]{KivlichanMcCleanEtAl2018}
I.~D. Kivlichan, J.~McClean, N.~Wiebe, C.~Gidney, A.~Aspuru-Guzik, G.~K.-L.
  Chan, and R.~Babbush.
\newblock Quantum simulation of electronic structure with linear depth and
  connectivity.
\newblock \emph{Phys. Rev. Lett.}, 120\penalty0 (11):\penalty0 110501, 2018.

\bibitem[Kivlichan et~al.(2020)Kivlichan, Gidney, Berry, Wiebe, McClean, Sun,
  Jiang, Rubin, Fowler, Aspuru-Guzik, et~al.]{kivlichan2020improved}
I.~D. Kivlichan, C.~Gidney, D.~W. Berry, N.~Wiebe, J.~McClean, W.~Sun,
  Z.~Jiang, N.~Rubin, A.~Fowler, A.~Aspuru-Guzik, et~al.
\newblock Improved fault-tolerant quantum simulation of condensed-phase
  correlated electrons via trotterization.
\newblock \emph{Quantum}, 4:\penalty0 296, 2020.
\newblock \doi{10.22331/q-2020-07-16-296}.

\bibitem[Knill et~al.(2007)Knill, Ortiz, and Somma]{knill2007optimal}
E.~Knill, G.~Ortiz, and R.~D. Somma.
\newblock Optimal quantum measurements of expectation values of observables.
\newblock \emph{Phys. Rev. A}, 75\penalty0 (1), 2007.
\newblock \doi{10.1103/PhysRevA.75.012328}.

\bibitem[Layden(2021)]{Layden2021first}
D.~Layden.
\newblock First-order {Trotter} error from a second-order perspective.
\newblock \emph{arXiv preprint arXiv:2107.08032}, 2021.

\bibitem[Lin and Tong(2020{\natexlab{a}})]{lin2020near}
L.~Lin and Y.~Tong.
\newblock Near-optimal ground state preparation.
\newblock \emph{Quantum}, 4:\penalty0 372, 2020{\natexlab{a}}.
\newblock \doi{10.22331/q-2020-12-14-372}.

\bibitem[Lin and Tong(2020{\natexlab{b}})]{lin2020optimal}
L.~Lin and Y.~Tong.
\newblock Optimal polynomial based quantum eigenstate filtering with
  application to solving quantum linear systems.
\newblock \emph{Quantum}, 4:\penalty0 361, 2020{\natexlab{b}}.
\newblock \doi{10.22331/q-2020-11-11-361}.

\bibitem[Low and Chuang(2017)]{low2017optimal}
G.~H. Low and I.~L. Chuang.
\newblock Optimal hamiltonian simulation by quantum signal processing.
\newblock \emph{Phys. Rev. Lett.}, 118\penalty0 (1):\penalty0 010501, 2017.
\newblock \doi{10.1103/physrevlett.118.010501}.

\bibitem[Low and Chuang(2019)]{low2019hamiltonian}
G.~H. Low and I.~L. Chuang.
\newblock Hamiltonian simulation by qubitization.
\newblock \emph{Quantum}, 3:\penalty0 163, 2019.
\newblock \doi{10.22331/q-2019-07-12-163}.

\bibitem[Lu et~al.(2020)Lu, Ba{\~n}uls, and Cirac]{LuBanulsCirac2020algorithms}
S.~Lu, M.~C. Ba{\~n}uls, and J.~I. Cirac.
\newblock Algorithms for quantum simulation at finite energies.
\newblock \emph{arXiv preprint arXiv:2006.03032}, 2020.

\bibitem[McArdle et~al.(2020)McArdle, Endo, Aspuru-Guzik, Benjamin, and
  Yuan]{McardleEtAl2020quantum}
S.~McArdle, S.~Endo, A.~Aspuru-Guzik, S.~C. Benjamin, and X.~Yuan.
\newblock Quantum computational chemistry.
\newblock \emph{Reviews of Modern Physics}, 92\penalty0 (1):\penalty0 015003,
  2020.

\bibitem[McClean et~al.(2016)McClean, Romero, Babbush, and
  Aspuru-Guzik]{mcclean2016theory}
J.~R. McClean, J.~Romero, R.~Babbush, and A.~Aspuru-Guzik.
\newblock The theory of variational hybrid quantum-classical algorithms.
\newblock \emph{New J. Phys.}, 18\penalty0 (2):\penalty0 023023, 2016.
\newblock \doi{10.1088/1367-2630/18/2/023023}.

\bibitem[Motta et~al.(2019)Motta, Sun, Tan, O'Rourke, Ye, Minnich,
  Brand{\~{a}}o, and Chan]{Motta2019QITE}
M.~Motta, C.~Sun, A.~T.~K. Tan, M.~J. O'Rourke, E.~Ye, A.~J. Minnich, F.~G.
  S.~L. Brand{\~{a}}o, and G.~K.-L. Chan.
\newblock Determining eigenstates and thermal states on a quantum computer
  using quantum imaginary time evolution.
\newblock \emph{Nature Phys.}, 16\penalty0 (2):\penalty0 205--210, 2019.
\newblock \doi{10.1038/s41567-019-0704-4}.

\bibitem[Nagaj et~al.(2009)Nagaj, Wocjan, and Zhang]{nagaj2009fast}
D.~Nagaj, P.~Wocjan, and Y.~Zhang.
\newblock Fast amplification of {QMA}.
\newblock \emph{Quantum Inf. Comput.}, 9\penalty0 (11):\penalty0 1053--1068,
  2009.

\bibitem[Nielsen and Chuang(2002)]{nielsen2002quantum}
M.~A. Nielsen and I.~Chuang.
\newblock Quantum computation and quantum information, 2002.

\bibitem[O'Brien et~al.(2020)O'Brien, Polla, Rubin, Huggins, McArdle, Boixo,
  McClean, and Babbush]{OBrienEtAl2020error}
T.~E. O'Brien, S.~Polla, N.~C. Rubin, W.~J. Huggins, S.~McArdle, S.~Boixo,
  J.~R. McClean, and R.~Babbush.
\newblock Error mitigation via verified phase estimation.
\newblock \emph{arXiv preprint arXiv:2010.02538}, 2020.

\bibitem[Oliveira and Terhal(2005)]{OliveiraTerhal2005}
R.~Oliveira and B.~M. Terhal.
\newblock The complexity of quantum spin systems on a two-dimensional square
  lattice.
\newblock \emph{arXiv preprint quant-ph/0504050}, 2005.

\bibitem[O’Brien et~al.(2019)O’Brien, Tarasinski, and
  Terhal]{OBrienTarasinskiTerhal2019quantum}
T.~E. O’Brien, B.~Tarasinski, and B.~M. Terhal.
\newblock Quantum phase estimation of multiple eigenvalues for small-scale
  (noisy) experiments.
\newblock \emph{New J. Phys.}, 21\penalty0 (2):\penalty0 023022, 2019.

\bibitem[O’Malley et~al.(2016)O’Malley, Babbush, Kivlichan, Romero,
  McClean, Barends, Kelly, Roushan, Tranter, Ding, et~al.]{omalley2016scalable}
P.~J. O’Malley, R.~Babbush, I.~D. Kivlichan, J.~Romero, J.~R. McClean,
  R.~Barends, J.~Kelly, P.~Roushan, A.~Tranter, N.~Ding, et~al.
\newblock Scalable quantum simulation of molecular energies.
\newblock \emph{Phys. Rev. X}, 6\penalty0 (3):\penalty0 031007, 2016.
\newblock \doi{10.1103/PhysRevX.6.031007}.

\bibitem[Parrish and McMahon(2019)]{parrish2019quantum}
R.~M. Parrish and P.~L. McMahon.
\newblock Quantum filter diagonalization: Quantum eigendecomposition without
  full quantum phase estimation.
\newblock \emph{arXiv preprint arXiv:1909.08925}, 2019.

\bibitem[Peruzzo et~al.(2014)Peruzzo, McClean, Shadbolt, Yung, Zhou, Love,
  Aspuru-Guzik, and O'Brien]{peruzzo2014variational}
A.~Peruzzo, J.~McClean, P.~Shadbolt, M.-H. Yung, X.-Q. Zhou, P.~J. Love,
  A.~Aspuru-Guzik, and J.~L. O'Brien.
\newblock A variational eigenvalue solver on a photonic quantum processor.
\newblock \emph{Nature Comm.}, 5\penalty0 (1), 2014.
\newblock \doi{10.1038/ncomms5213}.

\bibitem[Poulin and Wocjan(2009{\natexlab{a}})]{PoulinWocjan2009preparing}
D.~Poulin and P.~Wocjan.
\newblock Preparing ground states of quantum many-body systems on a quantum
  computer.
\newblock \emph{Phys. Rev. Lett.}, 102\penalty0 (13):\penalty0 130503,
  2009{\natexlab{a}}.
\newblock \doi{10.1103/PhysRevLett.102.130503}.

\bibitem[Poulin and Wocjan(2009{\natexlab{b}})]{poulin2009sampling}
D.~Poulin and P.~Wocjan.
\newblock Sampling from the thermal quantum {Gibbs} state and evaluating
  partition functions with a quantum computer.
\newblock \emph{Phys. Rev. Lett.}, 103\penalty0 (22), 2009{\natexlab{b}}.
\newblock \doi{10.1103/physrevlett.103.220502}.

\bibitem[Rhee and Glynn(2012)]{rhee2012new}
C.-H. Rhee and P.~W. Glynn.
\newblock A new approach to unbiased estimation for {SDE}'s.
\newblock In \emph{Proceedings of the 2012 Winter Simulation Conference (WSC)},
  pages 1--7. IEEE, 2012.
\newblock \doi{10.1109/WSC.2012.6465150}.

\bibitem[Rhee and Glynn(2015)]{rhee2015unbiased}
C.-H. Rhee and P.~W. Glynn.
\newblock Unbiased estimation with square root convergence for {SDE} models.
\newblock \emph{Oper. Res.}, 63\penalty0 (5):\penalty0 1026--1043, 2015.
\newblock \doi{10.1287/opre.2015.1404}.

\bibitem[Russo et~al.(2020)Russo, Rudinger, Morrison, and
  Baczewski]{RussoEtAl2020evaluating}
A.~Russo, K.~Rudinger, B.~Morrison, and A.~Baczewski.
\newblock Evaluating energy differences on a quantum computer with robust phase
  estimation.
\newblock \emph{arXiv preprint arXiv:2007.08697}, 2020.

\bibitem[Sanders et~al.(2020)Sanders, Berry, Costa, Tessler, Wiebe, Gidney,
  Neven, and Babbush]{SandersBerryEtAl2020compilation}
Y.~R. Sanders, D.~W. Berry, P.~C. Costa, L.~W. Tessler, N.~Wiebe, C.~Gidney,
  H.~Neven, and R.~Babbush.
\newblock Compilation of fault-tolerant quantum heuristics for combinatorial
  optimization.
\newblock \emph{PRX Quantum}, 1\penalty0 (2):\penalty0 020312, 2020.

\bibitem[Somma et~al.(2002)Somma, Ortiz, Gubernatis, Knill, and
  Laflamme]{SommaOrtizEtAl2002simulating}
R.~Somma, G.~Ortiz, J.~E. Gubernatis, E.~Knill, and R.~Laflamme.
\newblock Simulating physical phenomena by quantum networks.
\newblock \emph{Phys. Rev. A}, 65\penalty0 (4):\penalty0 042323, 2002.

\bibitem[Somma(2019)]{somma2019quantum}
R.~D. Somma.
\newblock Quantum eigenvalue estimation via time series analysis.
\newblock \emph{New J. Phys.}, 21\penalty0 (12):\penalty0 123025, 2019.
\newblock \doi{10.1088/1367-2630/ab5c60}.

\bibitem[Stair et~al.(2020)Stair, Huang, and
  Evangelista]{stair2020multireference}
N.~H. Stair, R.~Huang, and F.~A. Evangelista.
\newblock A multireference quantum {Krylov} algorithm for strongly correlated
  electrons.
\newblock \emph{J. Chem. Theory Comp.}, 16\penalty0 (4):\penalty0 2236--2245,
  feb 2020.
\newblock \doi{10.1021/acs.jctc.9b01125}.

\bibitem[Su et~al.(2020)Su, Huang, and Campbell]{SuHuangCampbell2020nearly}
Y.~Su, H.-Y. Huang, and E.~T. Campbell.
\newblock Nearly tight {Trotterization} of interacting electrons.
\newblock \emph{arXiv preprint arXiv:2012.09194}, 2020.

\bibitem[Sugisaki et~al.(2018)Sugisaki, Nakazawa, Toyota, Sato, Shiomi, and
  Takui]{SugisakiEtAl2018quantum}
K.~Sugisaki, S.~Nakazawa, K.~Toyota, K.~Sato, D.~Shiomi, and T.~Takui.
\newblock Quantum chemistry on quantum computers: A method for preparation of
  multiconfigurational wave functions on quantum computers without performing
  post-hartree--fock calculations.
\newblock \emph{ACS central science}, 5\penalty0 (1):\penalty0 167--175, 2018.

\bibitem[Suzuki(1991)]{Suzuki1991}
M.~Suzuki.
\newblock General theory of fractal path integrals with applications to
  many-body theories and statistical physics.
\newblock \emph{J. Math. Phys.}, 32\penalty0 (2):\penalty0 400--407, 1991.
\newblock \doi{10.1063/1.529425}.

\bibitem[Tran et~al.(2020)Tran, Chu, Su, Childs, and
  Gorshkov]{TranChuEtAl2020destructive}
M.~C. Tran, S.-K. Chu, Y.~Su, A.~M. Childs, and A.~V. Gorshkov.
\newblock Destructive error interference in product-formula lattice simulation.
\newblock \emph{Physical review letters}, 124\penalty0 (22):\penalty0 220502,
  2020.

\bibitem[Tubman et~al.(2018)Tubman, Mejuto-Zaera, Epstein, Hait, Levine,
  Huggins, Jiang, McClean, Babbush, Head-Gordon,
  et~al.]{TubmanEtAl2018postponing}
N.~M. Tubman, C.~Mejuto-Zaera, J.~M. Epstein, D.~Hait, D.~S. Levine,
  W.~Huggins, Z.~Jiang, J.~R. McClean, R.~Babbush, M.~Head-Gordon, et~al.
\newblock Postponing the orthogonality catastrophe: efficient state preparation
  for electronic structure simulations on quantum devices.
\newblock \emph{arXiv preprint arXiv:1809.05523}, 2018.

\bibitem[Wang et~al.(2019)Wang, Higgott, and Brierley]{wang2019accelerated}
D.~Wang, O.~Higgott, and S.~Brierley.
\newblock Accelerated variational quantum eigensolver.
\newblock \emph{Phys. Rev. Lett.}, 122\penalty0 (14):\penalty0 140504, 2019.
\newblock \doi{10.1103/physrevlett.122.140504}.

\bibitem[Wiebe et~al.(2015)Wiebe, Granade, Kapoor, and
  Svore]{wiebe2015bayesian}
N.~Wiebe, C.~Granade, A.~Kapoor, and K.~M. Svore.
\newblock Bayesian inference via rejection filtering.
\newblock \emph{arXiv preprint arXiv:1511.06458}, 2015.

\bibitem[Yi and Crosson(2021)]{YiCrosson2021spectral}
C.~Yi and E.~Crosson.
\newblock Spectral analysis of product formulas for quantum simulation.
\newblock \emph{arXiv preprint arXiv:2102.12655}, 2021.

\bibitem[Zwierz et~al.(2010)Zwierz, P{\'{e}}rez-Delgado, and
  Kok]{ZwierzPerezDelgadoKok2010}
M.~Zwierz, C.~A. P{\'{e}}rez-Delgado, and P.~Kok.
\newblock General optimality of the {Heisenberg} limit for quantum metrology.
\newblock \emph{Phys. Rev. Lett.}, 105\penalty0 (18), 2010.
\newblock \doi{10.1103/physrevlett.105.180402}.

\bibitem[Zwierz et~al.(2012)Zwierz, P{\'e}rez-Delgado, and
  Kok]{ZwierzPerezDelgadoKok2012ultimate}
M.~Zwierz, C.~A. P{\'e}rez-Delgado, and P.~Kok.
\newblock Ultimate limits to quantum metrology and the meaning of the
  {Heisenberg} limit.
\newblock \emph{Phys. Rev. A}, 85\penalty0 (4):\penalty0 042112, 2012.
\newblock \doi{10.1103/PhysRevA.85.042112}.

\end{thebibliography}

\appendix

\section{Constructing the approximate Heaviside function}
\label{sec:approx_heaviside}

In this appendix we construct the approximate Heaviside function satisfying the requirement in \eqref{eq:heaviside_approx_criterion_main_text}. We need to first construct a smeared Dirac function, which we will use as a mollifier in constructing the approximate Heaviside function. {To our best knowledge this particular version of smeared Dirac function has not been proposed in previous works.}

\begin{lem}
\label{lem:approx_dirac_function}
We define $M_{d,\delta}(x)$ by
\[
M_{d,\delta}(x) = \frac{1}{\mathcal{N}_{d,\delta}} T_d\left(1+2\frac{\cos(x)-\cos(\delta)}{1+\cos(\delta)}\right),
\]
where $T_d(x)$ is the $d$-th Chebyshev polynomial of the first kind, and 
\[
\mathcal{N}_{d,\delta} = \int_{-\pi}^{\pi} T_d\left(1+2\frac{\cos(x)-\cos(\delta)}{1+\cos(\delta)}\right)\dd x.
\]
Then
\begin{itemize}
    \item[(i)] $|M_{d,\delta}(x)|\leq \frac{1}{\mathcal{N}_{d,\delta}}$  for all $x\in[-\pi,-\delta]\cup[\delta,\pi]$, and $M_{d,\delta}(x)\geq -\frac{1}{\mathcal{N}_{d,\delta}}$ for all $x\in \RR$. %\LL{ check here} 
    %\LL{ also need $M_{d,\delta}(x)\geq -\frac{1}{\mathcal{N}_{d,\delta}}$  for all $x\in\RR$} \YT{Added}
    \item[(ii)]  $\int_{-\pi}^\pi M_{d,\delta}(x) \dd x = 1$, $1\leq \int_{-\pi}^\pi |M_{d,\delta}(x)| \dd x \leq 1+\frac{4\pi}{\mathcal{N}_{d,\delta}}$.
    \item[(iii)] When $\tan(\delta/2)\leq 1-1/\sqrt{2}$, we have
    \[
    \mathcal{N}_{d,\delta} \geq C_1 e^{d\delta/\sqrt{2}}\sqrt{\frac{\delta}{d}}\operatorname{erf}(C_2 \sqrt{d\delta})
    \]
    for some constants $C_1$ and $C_2$ that do not depend on $d$ or $\delta$. 
\end{itemize}
\end{lem}

\begin{proof}
We first note that, by the property of Chebyshev polynomials, when $x\in[-\pi,-\delta]\cup[\delta,\pi]$, i.e.  $\cos(x)\leq \cos(\delta)$, we have $\left|T_d\left(1+2\frac{\cos(x)-\cos(\delta)}{1+\cos(\delta)}\right)\right|\leq 1$. This proves the first inequality in (i). Note that when $x\in [-\delta,\delta]$, $T_d\left(1+2\frac{\cos(x)-\cos(\delta)}{1+\cos(\delta)}\right)\geq -1$. Combine this and the first inequality with the fact that $M_{d,\delta}(x)$ is $2\pi$-periodic we prove the second inequality in (i).

The first part of (ii) is obvious because of the definition of $\mathcal{N}_{d,\delta}$. For the second part, we have $\int_{-\pi}^{\pi}|M_{d,\delta}(x)|\dd x \geq \int_{-\pi}^{\pi}M_{d,\delta}(x)\dd x =1$. Also
\begin{equation}
\begin{aligned}
\int_{-\pi}^{\pi}|M_{d,\delta}(x)|\dd x &= \left(\int_{-\pi}^{-\delta} + \int_{\delta}^{\pi} \right) |M_{d,\delta}(x)| \dd x + \int_{-\delta}^{\delta}M_{d,\delta}(x) \dd x \\
&\leq \frac{4\pi}{\mathcal{N}_{d,\delta}}+\left(\int_{-\pi}^{-\delta} + \int_{\delta}^{\pi} \right) M_{d,\delta}(x) \dd x + \int_{-\delta}^{\delta}M_{d,\delta}(x) \dd x \\
&= 1 + \frac{4\pi}{\mathcal{N}_{d,\delta}}.
\end{aligned}
\label{eqn:M_abs_int_bound}
\end{equation}
We now prove (iii). This requires lower bounding  $T_d\left(1+2\frac{\cos(x)-\cos(\delta)}{1+\cos(\delta)}\right)$ when $x\in [-\delta,\delta]$. 
For $\delta$ small enough so that
\[
\max_{x}2\frac{\cos(x)-\cos(\delta)}{1+\cos(\delta)} = 2\tan^2(\delta/2)\leq 3-\sqrt{2},
\]
which is equivalent to $\tan(\delta/2)\leq 1-1/\sqrt{2}$,
we can use \cite[Lemma 13]{lin2020optimal} to provide a lower bound for the $x\in [-\delta,\delta]$ case:
\begin{equation}
\label{eq:Td_bound_1}
T_d\left(1+2\frac{\cos(x)-\cos(\delta)}{1+\cos(\delta)}\right) \geq \frac{1}{2}\exp\left(\sqrt{2}d\sqrt{\frac{\cos(x)-\cos(\delta)}{1+\cos(\delta)}}\right).
\end{equation}
By the elementary inequality
$
|\sin(x)|\leq |x|,
$
we have
% \[
% \frac{\cos(x)-\cos(\delta)}{1+\cos(\delta)} \geq \frac{1-\cos(\delta)}{1+\cos(\delta)}-\frac{x^2}{2+2\cos(\delta)}=\tan^2(\delta/2)\left(1-\frac{x^2}{4\sin^2(\delta/2)}\right).
% \]
% And
% \[
% \sqrt{\frac{\cos(x)-\cos(\delta)}{1+\cos(\delta)}} \geq \tan(\delta/2)\sqrt{1-\frac{x^2}{4\sin^2(\delta/2)}}\geq
% \tan(\delta/2)\left(1-\frac{x^2}{4\sin^2(\delta/2)}\right),
% \]
\[
\begin{aligned}
\sqrt{\frac{\cos(x)-\cos(\delta)}{1+\cos(\delta)}} &= \sqrt{\tan^2\left(\frac{\delta}{2}\right)-\frac{\sin^2(x/2)}{\cos^2(\delta/2)}} 
= \tan\left(\frac{\delta}{2}\right)\sqrt{1-\frac{\sin^2(x/2)}{\sin^2(\delta/2)}} \\
&\geq \tan\left(\frac{\delta}{2}\right)\left(1-\frac{\sin^2(x/2)}{\sin^2(\delta/2)}\right)
\geq \tan\left(\frac{\delta}{2}\right)\left(1-\frac{x^2}{4\sin^2(\delta/2)}\right).
\end{aligned}
\]
Substituting this into \eqref{eq:Td_bound_1} we have
\[
T_d\left(1+2\frac{\cos(x)-\cos(\delta)}{1+\cos(\delta)}\right) \geq
\frac{1}{2}e^{\sqrt{2}d\tan(\delta/2)}\exp\left(-\frac{dx^2}{\sqrt{2}\sin(\delta)}\right).
\]
Then %\LL{ is the factor $2$ there? please check} 
\begin{equation*}
% \label{eq:N_bound} 
\begin{aligned}
\mathcal{N}_{d,\delta} &\geq \int_{-\delta}^{\delta} T_d\left(1+2\frac{\cos(x)-\cos(\delta)}{1+\cos(\delta)}\right) \dd x -2\pi \\
&\geq \frac{1}{2}e^{\sqrt{2}d\tan(\delta/2)} \sqrt{\frac{\sqrt{2}\pi\sin(\delta)}{d}}\operatorname{erf}\left(\sqrt{\frac{d}{\sqrt{2}\sin(\delta)}}\delta\right) -2\pi\\
&\geq C_1 e^{d\delta/\sqrt{2}}\sqrt{\frac{\delta}{d}}\operatorname{erf}(C_2 \sqrt{d\delta}),
\end{aligned}
\end{equation*}
for $\delta\in(0,\pi/2)$ and some constants $C_1,C_2>0$. This proves (iii).

\end{proof}

\begin{figure}
    \centering
    \includegraphics[width=0.5\textwidth]{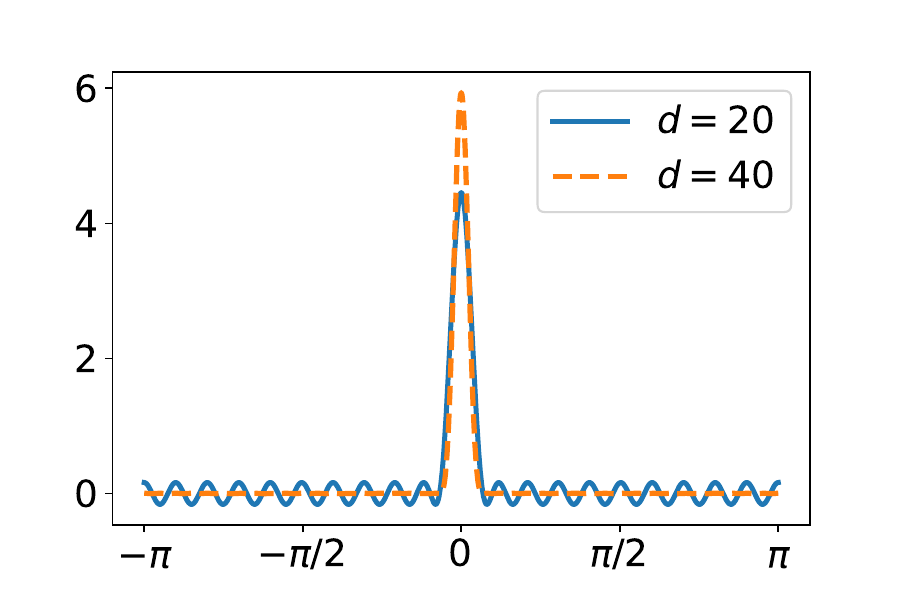}
    \caption{Illustration of $M_{d,\delta}(x)$ for $\delta=0.2$, $d=20,40$.}
    \label{fig:M_plot}
\end{figure}
A plot of $M_{d,\delta}$ is shown in Figure~\ref{fig:M_plot}. As we can see it roughly takes the shape of a Dirac function. We then use it as a mollifier to approximate the Heaviside function using the convolution of $M_{d,\delta}$ and the Heaviside function.

\begin{lem}
\label{lem:approx_heaviside_function}
Let $H(x)$ be the periodic Heaviside function defined in \eqref{eq:perodic_heaviside}.
% \[
% R(x) = \begin{cases}
% 1,&\ x\in(2k\pi,(2k+1)\pi), \\
% -1,&\ x\in((2k-1)\pi,2k\pi).
% \end{cases}
% \]
For any $\delta\in(0,\pi/2)$ such that $\tan(\delta/2)\leq 1-1/\sqrt{2}$ and $\epsilon>0$, there exists $d=\Or(\delta^{-1}\log(\delta^{-1}\epsilon^{-1}))$, and a $2\pi$-periodic function $F_{d,\delta}(x)$ of the form
\[
F_{d,\delta}(x) = \frac{1}{\sqrt{2\pi}}\sum_{k=-d}^d \hat{F}_{d,\delta,k}e^{ikx},
\]
satisfying
\begin{itemize}
    \item[(i)] $-\epsilon/2\leq  F_{d,\delta}(x) \leq 1+\epsilon$ for all $x\in\RR$;%\LL{ better state as $0\le F_{d,\delta}(x) \le 1$?} 
    \item[(ii)] $\vert F_{d,\delta}(x) - H(x) \vert\leq \epsilon$ for all $x\in [-\pi+\delta,-\delta]\cup[\delta,\pi-\delta]$;
    % \item[(iii)] $\sum_{k=-d}^d |\hat{F}_{d,\delta,k}| =\Or(\log(d))$.
    \item[(iii)] $ |\hat{F}_{d,\delta,k}| \leq 2(1+\epsilon)/(\sqrt{2\pi}|k|)$ for $k\neq 0$.
\end{itemize}
\end{lem}

\begin{proof}

We first construct the function $F_{d,\delta}(x)$.
Let $M_{d,\delta}(x)$ be the mollifier in Lemma~\ref{lem:approx_dirac_function}.
Because of Lemma~\ref{lem:approx_dirac_function} (i) and (ii)
$M_{d,\delta}(x)$ can be used as to mollify non-smooth functions. Also because $T_d(x)$ is a polynomial of degree $d$, the Fourier coefficients 
\[
\hat{M}_{d,\delta,k} = \frac{1}{\sqrt{2\pi}} \int_{-\pi}^{\pi} M_{d,\delta}(x) e^{-ikx}\dd x
\]
are non-zero only for $-d\leq k\leq d$. Also
\begin{equation}
\label{eq:M_hat_bound}
\left\vert\hat{M}_{d,\delta,k}\right\vert \leq \frac{1}{\sqrt{2\pi}} \int_{-\pi}^{\pi} |M_{d,\delta}(x)| \dd x=\frac{1+\epsilon}{\sqrt{2\pi}}.
\end{equation}
We construct $F_{d,\delta}$ by mollifying the Heaviside function with $M_{d,\delta}(x)$:
\begin{equation}
\label{eq:F_construction}
F_{d,\delta}(x) = (M_{d,\delta}* H)(x) = \int_{-\pi}^{\pi} M_{d,\delta}(x')H(x-x')\dd x'.
\end{equation}
% Because $\vert H(x)\vert\leq 1$, for any $x\in\RR$ we have
% $
% \left\vert F_{d,\delta}(x)\right\vert \leq 1.
% $
% Also since $M(x)\geq 0$, $H(x)\geq 0$ for all $x\in\RR$, $F_{d,\delta}(x)\geq 0$ for all $x\in \RR$.
% Therefore (i) is satisfied.

We then show we can choose $d=\Or(\delta^{-1}\log(\delta^{-1}\epsilon^{-1}))$ to satisfy (ii). We have
\[
\begin{aligned}
\left\vert F_{d,\delta}(x)- H(x) \right\vert &=\left\vert \int_{-\pi}^{\pi}M_{d,\delta}(x')(H(x-x')-H(x))\dd x' \right\vert \\
&\leq \int_{-\pi}^{\pi}M_{d,\delta}(x')\vert H(x-x')-H(x)\vert \dd x'.
\end{aligned}
\]
For any $x$ such that $|x|\in[\delta,\pi-\delta]$, first we consider the case where $|x'|< \delta$. In this case $H(x-x')=H(x)$ and therefore the integrand $M_{d,\delta}(x')\vert H(x-x')-H(x)\vert=0$. Then we consider the case  where $|x'|\geq \delta$. By Lemma~\ref{lem:approx_dirac_function} (i) we have $M_{d,\delta}(x')\leq 2/\mathcal{N}_{d,\delta}$, and as $\vert H(x-x')-H(x)\vert\leq 1$,  $M_{d,\delta}(x')\vert H(x-x')-H(x)\vert\leq 2/\mathcal{N}_{d,\delta}$. Thus for any $x$ such that $|x|\in[\delta,\pi-\delta]$,
\begin{equation}
\label{eq:fourier_approx_error}
\left\vert F_{d,\delta}(x)- H(x) \right\vert \leq \frac{4\pi}{\mathcal{N}_{d,\delta}}.
\end{equation}

If we want to keep the approximation error for $x\in[-\pi+\delta,-\delta]\cup[\delta,\pi-\delta]$ to be below $\epsilon$, we will need, by Lemma~\ref{lem:approx_dirac_function} (i) and \eqref{eq:fourier_approx_error},
\[
C_1 e^{d\delta/\sqrt{2}}\sqrt{\frac{\delta}{d}}\operatorname{erf}(C_2 \sqrt{d\delta}) \geq \frac{4\pi}{\epsilon}.
\]
It can be checked that we can choose
$
d = \Or(\delta^{-1}\log(\epsilon^{-1}\delta^{-1}))
$
to achieve this. 

We then show this choice of $d$ ensures (i) as well. From \cref{eqn:M_abs_int_bound}
\[
F_{d,\delta}(x) \leq \int_{-\pi}^{\pi} |M_{d,\delta}(y)|\dd y \leq 1+\frac{4\pi}{\mathcal{N}_{d,\delta}}\leq 1+\epsilon 
\]
and by the second inequality in Lemma~\ref{lem:approx_dirac_function} (i)%\LL{i think you need the global lower bound for $M$ here } 
\[
F_{d,\delta}(x) \geq -\frac{1}{\mathcal{N}_{d,\delta}}\int_{-\pi}^{\pi} H(y)\dd y = -\frac{2\pi}{\mathcal{N}_{d,\delta}}\geq -\frac{\epsilon}{2}.
\]

Finally we prove our construction satisfies (iii).
Because $F_{d,\delta}(x)$ is defined through a convolution, its Fourier coefficients can be obtained through
\[
\hat{F}_{d,\delta,k} = \sqrt{2\pi}\hat{M}_{d,\delta,k}\hat{H}_k,
\]
where $\hat{H}_k$'s are the Fourier coefficients of the rectangle function $H(x)$.
Therefore $\hat{F}_{d,\delta,k}\neq 0$ only for $-d\leq k\leq d$. 
% The sum $\sum_{k=-d}^d|\hat{F}_{d,\delta,k}|$ is of interest in our algorithm. 
Because of \eqref{eq:M_hat_bound}, we have
\[
|\hat{F}_{d,\delta,k}| \leq  (1+\epsilon)|\hat{H}_{k}|.
\]
Since when $k\neq 0$
\[
\begin{aligned}
\hat{H}_{k} &= \frac{1}{\sqrt{2\pi}}\int_{-\pi}^{\pi}H(x)e^{-ikx}\dd x \\
&=\begin{cases}
\frac{2}{i\sqrt{2\pi}k} &\ 2\nmid k \\
0 &\ 2\mid k
\end{cases}
\end{aligned}
\]
we have (iii).
% \[
% \sum_{k=-d}^d|\hat{F}_{d,\delta,k}| \leq  \sum_{k=-d}^d\frac{4}{\sqrt{2\pi}|k|}=\Or(\log(d)).
% \]

\end{proof}

\section{The relation between the CDF and the approximate CDF}
\label{sec:relation_CDF_ACDF}

In this appendix we prove \eqref{eq:relation_CDF_ACDF}. Let $0<\delta<\pi/6$.
First we have a $2\pi$-periodic function $F(x)$ from Lemma~\ref{lem:approx_heaviside_function} that satisfies
% \begin{equation}
% \label{eq:heaviside_approx_fourier}
%     \wt{H}(x) = \sum_{j=-d}^d \wt{H}_{j}e^{ij x}
% \end{equation} 
% that approximates the Heaviside function $H(x)$ in the following sense:
\begin{equation*}
    % \label{eq:heaviside_approx}
    |F(x)-H(x)|\leq \epsilon,\quad  x\in [-\pi+\delta,-\delta]\cup[\delta,\pi-\delta],
\end{equation*}
and $F(x)\in[0, 1]$ for all $x\in\RR$. We further define $F_L(x)=F(x-\delta)$ and $F_R(x)=F(x+\delta)$. They satisfy
\begin{equation}
    \label{eq:H_shifted_bound}
    \begin{aligned}
    |F_L(x)-H(x)|\leq \epsilon,&\quad x\in [-\pi+2\delta,0]\cup[2\delta,\pi], \\
    |F_R(x)-H(x)|\leq \epsilon,&\quad x\in [-\pi,-2\delta]\cup[0,\pi-2\delta].
    \end{aligned}
\end{equation}
% Furthermore $|F_L(x)-H(x)|\leq 1$ and $|F_R(x)-H(x)|\leq$ for all $x$.

We define the some functions related to the ACDF as follows:
\begin{equation}
    \label{eq:CLR}
    \wt{C}_L(x)=(F_L*p)(x),\quad \wt{C}_R(x)=(F_R*p)(x).
\end{equation}
Then we have
\begin{equation}
    \label{eq:shift_relation}
    \wt{C}_L(x)=\wt{C}(x-\delta),\quad \wt{C}_R(x)=\wt{C}(x+\delta).
\end{equation}
The functions $\wt{C}_L(x)$ and $\wt{C}_R(x)$ can be used to bound $C(x)$. Because of \eqref{eq:H_shifted_bound}, the fact that $p(x)$ is supported in $(-\pi/3,\pi/3)$ in $[-\pi,\pi]$, $\delta<\pi/6$, and that $H(y)$ and $F_L(y)$ both take value in $[0,1]$, for $x\in(-\pi/3,\pi/3)$ we have %\LL{ you also need $|H(y)-F_L(y)|\le 1$ on $[0,2\delta]$} 
\[
\begin{aligned}
|\wt{C}_L(x)-C(x)| &\leq \int_{-\pi}^{\pi}p(x-y)|H(y)-F_L(y)|\dd y \\
&\leq \epsilon+\int_{0}^{2\delta}p(x-y)\dd y \\
&=\epsilon+C(x)-C(x-2\delta).
\end{aligned}
\]
Therefore
\begin{equation*}
    \label{eq:C_bound_left}
    \wt{C}_L(x) \geq C(x)-[\epsilon+C(x)-C(x-2\delta)] = C(x-2\delta)-\epsilon.
\end{equation*}
Similarly we have
\begin{equation*}
    \label{eq:C_bound_right}
    \wt{C}_R(x) \leq C(x)+[\epsilon+C(x+2\delta)-C(x)] = C(x+2\delta)+\epsilon.
\end{equation*}
Combining these two inequalities with \eqref{eq:shift_relation}, we have
\begin{equation*}
    % \label{eq:C_bound}
    C(x-2\delta)\leq \wt{C}(x-\delta)+\epsilon,\quad C(x+2\delta)\geq \wt{C}(x+\delta)-\epsilon.
\end{equation*}
This proves \eqref{eq:relation_CDF_ACDF}.

\section{Obtaining the ground state energy by solving the QEEP}
\label{sec:qeep}

Here we discuss how to obtain the ground state energy using algorithm in Ref.~\cite{somma2019quantum} to solve the QEEP. The cost of solving the QEEP as analyzed in Ref.~\cite{somma2019quantum} scales as $\epsilon^{-6}$. However, the cost can be much reduced for the problem of ground state energy estimation. For simplicity we assume $\|H\|<\pi/3$ and $\tau$ is chosen to be $1$.

In order to find the interval of size $2\epsilon$ containing the ground state energy , we first divide the interval $[-\pi/3,\pi/3]$ into $M$ bins of equal size smaller than $2\epsilon$. We then define the indicator function associated with an interval $[a,b]$ to be
\[
{1}_{[a,b]}(x) = \begin{cases}
1,&\ x\in[a,b],\\
0,&\ x\notin [a,b].
\end{cases}
\]
In QEEP the goal is to estimate $\Tr[\rho {1}_{[a,b]}(H)]$, where $[a,b]$ is one of the $M$ bins, to within precision $\Or(\epsilon)$. However, in our setting, if we know $p_0\geq \eta$, one can estimate $\Tr[\rho {1}_{[a,b]}(H)]$ to within error $\Or(\eta)$. If we get $\Tr[\rho {1}_{[a,b]}(H)]<\eta$ with high confidence then we know the ground state energy $\lambda_0$ is not in this interval. If know $\Tr[\rho {1}_{[a,b]}(H)]>\eta/2$ with high confidence then there is an eigenvalue in $[a,b]$. If the above task can be done, then we choose the leftmost bin in which $\Tr[\rho {1}_{[a,b]}(H)]>\eta/2$. This will enable us to solve the ground state energy estimation problem.

To estimate $\Tr[\rho {1}_{[a,b]}(H)]$, Ref.~\cite{somma2019quantum} first approximated the indicator function $1_{[a,b]}(x)$ using a truncated Fourier series \cite[Appendix A]{somma2019quantum}, similar to what we did in \cref{sec:approx_heaviside}. The number of terms $N_{\mathrm{term}}$ and the maximal evolution time $T$ both scale like $\epsilon^{-1}$. In Ref.~\cite{somma2019quantum} the author proposed estimating each Fourier mode $\Tr[\rho e^{-ijH}]$ to within error $\Or(\epsilon/N_{\mathrm{term}})$. Because here the estimation precision is $\Or(\eta)$ rather than $\Or(\epsilon)$, we should instead estimate $\Tr[\rho e^{-ijH}]$ to within error $\Or(\eta/N_{\mathrm{term}})=\Or(\eta\epsilon)$. Because we are using Monte Carlo sampling this requires $\Or(\eta^{-2}\epsilon^{-2})$ samples. We need the same number of samples for each $\Tr[\rho e^{-ijH}]$, and therefore the total time we need to run time evolution is $\Or(N_{\mathrm{term}} T \eta^{-2}\epsilon^{-2})=\Or(\eta^{-2}\epsilon^{-4})$. We omitted polylogarithmic factors in the complexity.

However if the analysis is done more carefully the dependence on $\epsilon$ could be improved. First one should notice that the error for each $\Tr[\rho e^{-ijH}]$ is independent, and the estimate is unbiased (if we do not consider the Fourier approximation error), as is the case in our algorithm (Section~\ref{sec:evaluate_ACDF}). Therefore the total error for estimating $\Tr[\rho {1}_{[a,b]}(H)]$ accumulates sublinearly. {More precisely, let the error for estimating $\Tr[\rho e^{-ijH}]$ be $\varepsilon_j$ with variance $\sigma_j^2$, and let the coefficient for $\Tr[\rho e^{-ijH}]$ be $A_j$. Then the total error $\sum_j A_j \varepsilon_j$ has variance $\sum_j A_j^2 \sigma_j^2$. Therefore the total error is roughly $\sqrt{\sum_j A_j^2 \sigma_j^2}$ instead of the linearly accumulated error $\sum_j A_j\sigma_j$. These two can have different asymptotic scaling depending on the magnitude of $A_j$.} {Because of this one} can in fact choose to estimate $\Tr[\rho e^{-ijH}]$ to within error $\Or(\eta/\sqrt{N_{\mathrm{term}}})=\Or(\eta\epsilon^{-1/2})$. This saves a $\epsilon^{-1}$ factor in the total runtime. Furthermore, one can choose to evaluate the approximate indicator function in a stochastic way, like we did in Section~\ref{sec:evaluate_ACDF}. By taking into account the decay of Fourier coefficients, similar to Lemma~\ref{lem:approx_heaviside_function} (iii), it is possible to further reduce the complexity.

\section{Complexity analysis for using Trotter formulas}
\label{sec:trotter}

In this appendix, instead of using the maximal evolution time and the total evolution time to quantify the complexity, we directly analyze the circuit depth and the total runtime when the time evolution is simulated using Trotter formulas.
We suppose the Hamiltonian $H$ can be decomposed as $H=\sum_{\gamma} H_\gamma$, where each of $H_{\gamma}$ can be efficiently exponentiated. A $p$-th order Trotter formula applied to $e^{-i\tau H}$ with $r$ Trotter steps gives us a unitary operator $U_{\mathrm{HS}}$ with error
\[
\|U_{\mathrm{HS}}-e^{-i\tau H}\|\leq C_{\mathrm{Trotter}}\tau^{p+1} r^{-p},
\]
where $C_{\mathrm{Trotter}}$ is a prefactor, for which the simplest bound is $C_{\mathrm{Trotter}}=\Or((\sum_{\gamma}\|H\|_{\gamma})^{p+1})$. Tighter bounds in the form of a sum of commutators are proved in Refs.~\cite{ChildsSuTranWiebeZhu2021,SuHuangCampbell2020nearly}.

\subsection{The algorithm in this work}
\label{sec:trotter_this_work}

Our algorithm requires approximating Eq.~\eqref{eq:evaluate_ACDF} to precision $\eta$ (as in Theorem~\ref{cor:ground_energy} $\eta$ is a lower bound of $p_0/2$) using Trotter formulas. Suppose we are using a $p$-th order Trotter formula, then we want
\[
\Big\|\sum_{j}\hat{F}_j e^{ijx} \Tr[\rho e^{-ij\tau H}] - \sum_{j}\hat{F}_j e^{ijx} \Tr[\rho U_{\mathrm{HS}}^j]\Big\|=\Or(\eta).
\]
Since the left-hand side can be upper bounded by
\[
\sum_{j}|\hat{F}_j||j|  \|e^{-i\tau H} - U_{\mathrm{HS}}\| = \Or(d\|e^{-i\tau H} - U_{\mathrm{HS}}\|)
\]
by Lemma~\ref{lem:approx_heaviside_function} (iii), we only need to choose $r$ so that
\[
C_{\mathrm{Trotter}}\tau^{p+1} r^{-p} = \Or(\eta d^{-1}).
\]
Therefore we can choose
\[
r = \max\{1, \wt{\Or}(d^{1/p}\eta^{-1/p}C_{\mathrm{Trotter}}^{1/p}\tau^{1+1/p})\}
\]
The maximal evolution time in Corollary~\ref{cor:ground_energy} tells us how many times we need to use the operator $U_{\mathrm{HS}}$ (multiplied by a factor $\tau$). Multiply this by $r$ we have the maximal circuit depth we need, which is
\begin{equation}
\label{eq:circuit_depth_our_work}
dr=\wt{\Or}(\max\{\tau^{-1}\epsilon^{-1},\epsilon^{-1-1/p}\eta^{-1/p}C^{1/p}_{\mathrm{Trotter}}\}).
\end{equation}
Similarly we have the total runtime
\begin{equation}
\label{eq:total_runtime_our_work}
\wt{\Or}(\max\{\tau^{-1}\epsilon^{-1}\eta^{-2},\epsilon^{-1-1/p}\eta^{-2-1/p}C^{1/p}_{\mathrm{Trotter}}\}).
\end{equation}
If we fix $H$ and let $\epsilon,\eta\to 0$, then we can see this gives us an extra $\epsilon^{-1/p}\eta^{-1/p}$ factor in the circuit depth and total runtime, compared to the maximal evolution time and the total evolution time respectively.

\subsection{Quantum phase estimation}
\label{sec:trotter_qpe}

We then analyze the circuit depth and total runtime requirement for estimating the ground state energy with QPE, where the time evolution is performed using Trotter formulas. We analyze the multi-ancilla qubit version of QPE and the result is equally valid for the single-ancilla qubit version using semi-classical Fourier transform.

In QPE, when we replace all exact time evolution with $U_{\mathrm{HS}}$,
we would like to ensure that the probability of obtaining an energy measurement close to the ground state energy remains bounded away from $0$ by $\Omega(\eta)$.
Therefore the probability distribution of the final measurement outcome should be at most $\Or(\eta)$ away from the original distribution in terms of the total variation distance. 

Because the only part of QPE that depends on the time evolution operator is the multiply-controlled unitary
\[
\sum_{j=0}^{J-1} \ket{j}\bra{j}\otimes e^{-ij\tau H},
\]
which is replaced by
\[
\sum_{j=0}^{J-1} \ket{j}\bra{j}\otimes U_{\mathrm{HS}}^j
\]
when we use Trotter formulas, we only need to ensure the difference between the two operators to be upper bounded by $\Or(\eta)$ in terms of operator norm. Therefore we need
\[
J\|e^{-ij\tau H}-U_{\mathrm{HS}}\|=\Or(\eta).
\]
As discussed in Section~\ref{sec:related}, we need to choose $J=\Or(\tau^{-1}\epsilon^{-1}\eta^{-1})$ (we need the $\tau^{-1}$ factor to account for rescaling $H$, and $p_0$ in Section~\ref{sec:related} is replaced by $\eta$). Following the same analysis as in the previous section, we need to choose the number of Trotter steps for approximating $e^{-i\tau H}$ to be
\[
r=\max\{1,\Or(J^{1/p}\eta^{-1/p}C_{\mathrm{Trotter}}^{1/p}\tau^{1+1/p})\}
\]
Therefore the circuit depth needed is
\begin{equation}
\label{eq:circuit_depth_qpe}
Jr = \Or(\max\{\tau^{-1}\epsilon^{-1}\eta^{-1},\epsilon^{-1-1/p}\eta^{-1-2/p}C_{\mathrm{Trotter}}^{1/p}\}),
\end{equation}
and the total runtime is
\begin{equation}
\label{eq:total_runtime_qpe}
\Or(\max\{\tau^{-1}\epsilon^{-1}\eta^{-2},\epsilon^{-1-1/p}\eta^{-2-2/p}C_{\mathrm{Trotter}}^{1/p}\}).
\end{equation}
Again, if we fix $H$ and let $\epsilon,\eta\to 0$, then we can see this gives us an extra $\epsilon^{-1/p}\eta^{-2/p}$ factor in the circuit depth and total runtime, compared to the maximal evolution time and the total evolution time respectively. This is worse by a factor of $\eta^{-1/p}$ than the cost using our algorithm.

\section{The control-free setting}
\label{sec:control_free}

In this appendix we introduce, as an alternative to the quantum circuit in \eqref{eq:Hadamard_circuit}, a circuit which does not require controlled time evolution. This construction is mainly based on the ideas in Refs.~\cite{RussoEtAl2020evaluating,LuBanulsCirac2020algorithms,OBrienEtAl2020error}. We will introduce the construction of the circuit and discuss how to use the measurement results from the circuit to construct a random variable $\wt{Z}$ satisfying 
\begin{equation}
\label{eq:wtZ_requirement}
\mathbb{E}[\wt{Z}] = \Tr[\rho e^{-itH}]
\end{equation}
for any given $t$. Then choosing $t=j\tau$, we will be able to replace $X_j$ and $Y_j$ with $\Re \wt{Z}$ and $\Im \wt{Z}$ respectively, while satisfying \eqref{eq:measurement_outcome_expect_X} and \eqref{eq:measurement_outcome_expect_Y}.
In order to remove the need of performing controlled time evolution of $H$, we need some additional assumptions.
\begin{enumerate}
    \item The initial state $\rho$ is a pure state $\ket{\phi_0}$, prepared using a unitary circuit $U_I$.
    \item We have a reference eigenstate $\ket{\psi_{R}}$ of $H$ corresponding to a known eigenvalue $\lambda_{R}$. This eigenstate can be efficiently prepared using a unitary circuit $U_R$.
    \item $\braket{\psi_{R}|\phi_0}=0$. %\LL{ slightly more natural to  say  $\braket{\psi_{R}|\phi_0}=0$?}
\end{enumerate}
The last assumption $\braket{\psi_{R}|\phi_0}=0$ implies $\braket{\psi_{R}|e^{-itH}|\phi_0}=0$ for all $t\in\RR$ because $\ket{\psi_R}$ is an eigenvector of $e^{-itH}$.
All of these are reasonable assumptions for a second-quantized fermionic Hamiltonian: we choose $\ket{\psi_{R}}$ to be the vacuum state, $\lambda_{R}=0$, and $\ket{\phi_0}$ to be the Hartree-Fock state, which can be efficiently prepared \cite{KivlichanMcCleanEtAl2018}. Naturally $\braket{\psi_{R}|\phi_0}=0$ because of the particle number conservation.

With these assumptions, we let 
\[
\alpha=\braket{\phi_0|e^{-it(H-\lambda_{R})}|\phi_0}.
\]
Also define
\[
\ket{\Psi_{0,\pm}}=\frac{1}{\sqrt{2}}(\ket{\psi_R}\pm\ket{\phi_0}),\quad
\ket{\Psi_{1,\pm}}=\frac{1}{\sqrt{2}}(\ket{\psi_R}\pm i\ket{\phi_0}).
\]
With these states, we can express $\alpha$ in terms of expectation values:
\[
\begin{aligned}
\braket{\Psi_{0,+}|e^{-itH}|\Psi_{0,\pm}} &= \frac{1}{2}e^{-i\lambda_{R}t}(1\pm\alpha), \\
\braket{\Psi_{0,+}|e^{-itH}|\Psi_{1,\pm}} &= \frac{1}{2}e^{-i\lambda_{R}t}(1\pm i\alpha).
\end{aligned}
\]
In Refs.~\cite{LuBanulsCirac2020algorithms,RussoEtAl2020evaluating} it is assumed that we have unitary circuits to prepare $\ket{\Psi_{0,\pm}}$ and  $\ket{\Psi_{1,\pm}}$. However it is not immediately clear how these circuits are constructed. Here we will take a slightly different approach. The circuit diagram is as follows:
\begin{equation}
\label{eq:circuit_alt_Hadamard}
\begin{quantikz}
\lstick{$\ket{0}$}   & \gate{\mathrm{H}} & \gate{K} & \ctrl{2}   & \octrl{2}  &  
                 \qw & \qw                   & \qw                 & \gate{\mathrm{H}} & \meter{} \\
\lstick{$\ket{0}$}   & \gate{\mathrm{H}} & \qw      & \qw        & \qw        &   
                 \qw & \ctrl{1}              & \octrl{1}           & \gate{\mathrm{H}} & \meter{} \\
\lstick{$\ket{0^n}$} & \qwbundle[]{}     & \qw      & \gate{U_I} & \gate{U_R} &
    \gate{e^{-itH}}  & \gate{U_I^{\dagger}}  & \gate{U_R^{\dagger}}& \qw               & \meter{}\\
\end{quantikz}
\end{equation}
In this circuit we choose $K=I$ for the real part of $\alpha$ or the phase gate $S$ for the imaginary part of $\alpha$. This circuit uses three registers, with the first two containing one qubit each, and the third one containing $n$ qubits. 

We first analyze the probability of different measurement outcomes when $K=I$. When we run the above circuit, and measure all the qubits, the probability of the measurement outcomes of the first two qubits being $(b_1,b_2)$, and the rest of the qubits being all $0$, is
\[
\begin{aligned}
p_{0,(b_1,b_2)} &= 
\begin{cases}
|\braket{\Psi_{0,+}|e^{-itH}|\Psi_{0,+}}|^2/4,&\ b_1=b_2 \\
|\braket{\Psi_{0,+}|e^{-itH}|\Psi_{0,-}}|^2/4,&\ b_1\neq b_2 \\
\end{cases} \\
&= \frac{1}{16}(1+|\alpha|^2+ 2(-1)^{b_1+b_2}\Re \alpha).
\end{aligned}
\]
Here we have used the fact that $|\braket{\Psi_{0,+}|e^{-itH}|\Psi_{0,+}}|=|\braket{\Psi_{0,-}|e^{-itH}|\Psi_{0,-}}|$. 

Similarly, when $K=S$, the probability of the measurement outcomes of the first two qubits being $(b_1,b_2)$, and the rest of the qubits being all $0$, is
\[
\begin{aligned}
p_{1,(b_1,b_2)} &= 
\begin{cases}
|\braket{\Psi_{0,+}|e^{-itH}|\Psi_{1,+}}|^2/4,&\ b_1=b_2 \\
|\braket{\Psi_{0,+}|e^{-itH}|\Psi_{1,-}}|^2/4,&\ b_1\neq b_2 \\
\end{cases} \\
&= \frac{1}{16}(1+|\alpha|^2- 2(-1)^{b_1+b_2}\Im \alpha).
\end{aligned}
\]

Based on the above analysis, we construct the random variable $\wt{Z}$ in the following way: we first run the circuit with $K=I$, and denote the measurement outcomes of the first two qubits by $(b_1,b_2)$. If the third register returns all $0$ when measured, then we let $\wt{X}=(-1)^{b_1+b_2}$. Otherwise we let $\wt{X}=0$. Similarly we define a random variable $\wt{Y}$ for $K=S$. We have
\[
\mathbb{E}[\wt{X}] = p_{0,(0,0)} + p_{0,(1,1)} - p_{0,(0,1)} - p_{0,(1,0)} = \frac{1}{2}\Re\alpha,
\]
and
\[
\mathbb{E}[\wt{Y}] = p_{1,(0,0)} + p_{1,(1,1)} - p_{1,(0,1)} - p_{1,(1,0)} = -\frac{1}{2}\Im\alpha.
\]
Therefore we can define
\[
\wt{Z} = 2e^{-i\lambda_R t}(\wt{X}-i\wt{Y}).
\]
Then 
\[
\mathbb{E}[\wt{Z}] = e^{-i\lambda_R t}\alpha = \Tr[\rho e^{-itH}].
\]
Thus we can see this new random variable $\wt{Z}$ satisfies \eqref{eq:wtZ_requirement}. Compared to the $Z$ in the main text this new random variable has a slightly larger variance:
\[
\mathrm{var}[\wt{Z}] \leq \mathbb{E}[|\wt{Z}|^2]\leq 8.
\]
This however does not change the asymptotic complexity.

\section{Details on the numerical experiments}
\label{sec:detail_numerical}

In Figure~\ref{fig:ACDF},
we apply the procedure described in Section~\ref{sec:evaluate_ACDF} to approximate the CDF of the Fermi-Hubbard model, described by the Hamiltonian
\begin{equation}
\label{eq:hubbard_ham}
H = -t\sum_{\braket{j,j'},\sigma} c_{j,\sigma}^{\dagger}c_{j',\sigma} + U\sum_{j} \left(n_{j,\uparrow}-\frac{1}{2}\right)\left(n_{j,\downarrow}-\frac{1}{2}\right),
\end{equation}
where $c_{j,\sigma}$ ($c^{\dag}_{j,\sigma}$) denotes the fermionic annihilation (creation) operator on the site $j$ with spin $\sigma\in\{\uparrow,\downarrow\}$. $\braket{\cdot,\cdot}$ denotes sites that are adjacent to each other. $n_{j,\sigma}=c_{j,\sigma}^{\dagger}c_{j,\sigma}$ is the number operator. The sites are arranged into a one-dimensional chain, with open boundary condition. 

We first evaluate $\bar{G}(x)$ defined in \eqref{eq:G_average}, and the result is shown in Figure~\ref{fig:ACDF}. We use a classical computer to simulate the sampling from the quantum circuit. The initial state $\rho$ is chosen to be the Hartree-Fock solution, which has an overlap of around $0.4$ with the exact ground state. We can see that $\bar{G}(x)$ closely follows the CDF, and even though there is significant noise from Monte Carlo sampling, the jump corresponding to the ground state energy is clearly resolved.

Then we consider estimating the ground state energy from $\bar{G}(x)$. In this numerical experiment we use a heuristic approach, and the rigorous approach that comes with provable error bound and confidence level is discussed in Sections~\ref{sec:find_ground_energy} and \ref{sec:invert_CDF}. We obtain the estimate by
\[
x^{\star} = \inf\{x:\bar{G}(x)\geq {\eta/2}\},
\]
and $x^{\star}/\tau$ is an estimate for the ground state energy $\lambda_0$. We expect $x^{\star}\in [\tau\lambda_0-\delta,\tau\lambda_0+\delta]$. {Here $\eta$ is chosen so that $p_0\geq \eta$.}%In the above equation we used {$\eta/2$} but in practice we can replace this by any number between $0$ and $p_0$, as long as it is sufficiently far away from $0$ and $p_0$ to account for the approximation and sampling errors. 

The error of the estimated ground state energy,  the total evolution time, and the maximal evolution time are shown in Figure~\ref{fig:error_and_time}, in which we have chosen $U/t=4$ for the Hubbard model. In the right panel of Figure~\ref{fig:error_and_time} we can see the line for total evolution time runs parallel to the line for the maximal evolution time. Because the maximal evolution time scales linearly with respect to $\delta^{-1}$, and this plot uses logarithmic scales for both axes, we can see the total evolution time has a $\delta^{-1}$ scaling, and is therefore inversely proportional to the allowed error of ground state energy estimation.

\section{Frequently used symbols}
\label{sec:frequently_used_symbols}

\begin{table}[h!]
    \centering
    \makegapedcells
    {
    \begin{tabular}{c|c}
     \hline
      \hline
        Symbol & Meaning \\
        \hline
        $H$ & The Hamiltonian for which we want to estimate the ground state energy. \\
         \hline
        $\rho$ & The initial state from which we perform time evolution and measurement. \\
         \hline
        $p_k$ & The overlap between $\rho$ and the $k$-th lowest eigensubspace. \\
         \hline
        $\tau$ & A renormalization factor satisfying $\tau\|H\|\leq \pi/4$. \\
         \hline
        $p(x)$ & The spectral density associated with $\tau H$ and $\rho$. \\
         \hline
        $C(x)$ & The cumulative distribution function defined in \eqref{eq:CDF_def}. \\
         \hline
        $\wt{C}(x)$ & The approximate CDF defined in \eqref{eq:ACDF_definition}. \\
         \hline
        $G(x)$ & An unbiased estimate of the ACDF $\wt{C}(x)$ defined in \eqref{eq:unbiased_estimate_for_ACDF}. \\
         \hline
        $\bar{G}(x)$ &The average of multiple samples of $G(x)$, defined in \eqref{eq:G_average}. \\
         \hline
        $J_k$ &  \makecell{An integer drawn from the distribution \eqref{eq:J_dist} \\ signifying the number of steps in the time evolution. $|J_k|\leq d$.} \\
         \hline
        $Z_k$ & \makecell{A sample generated on a quantum circuit from two measurement outcomes. \\ Defined in \eqref{eq:def_Z}. Can only take value $\pm 1\pm i$.} \\
        \hline
        $d$  & The maximal possible value of $|J_k|$. \\
        \hline
        $\delta$ & \makecell{In the context of Corollary~\ref{cor:ground_energy} we choose $\delta=\tau\epsilon$ \\ where $\epsilon$ is the allowed error of the ground state energy.} \\
        \hline
        $\vartheta$ & The allowed failure probability. \\
         \hline
          \hline
    \end{tabular}
    }
    \caption{Frequently used symbols in this work.}
    \label{tab:frequently_used_symbols}
\end{table}

\end{document}